\documentclass[runningheads, 11pt]{llncs}
\usepackage[left=1in,top=1in,right=1in,bottom=1in]{geometry}

\usepackage[T1]{fontenc}
\usepackage{amsmath,amssymb,amsfonts}
\usepackage{xcolor}
\usepackage{accents}
\usepackage{graphicx}
\usepackage{algorithm,algorithmic}
\usepackage{hyperref}
\usepackage{color}

\usepackage{subcaption}
\usepackage[normalem]{ulem}
\usepackage{dsfont}
\usepackage{mathtools}
\usepackage{array}
\usepackage{tabularx}
\newcolumntype{Y}{>{\centering\arraybackslash}X}
\usepackage{makecell}

\usepackage{enumitem}
\usepackage{pifont}

\newcommand{\xmark}{\ding{55}}
\usepackage{multirow}
\usepackage{bm}

\usepackage[suppress]{color-edits} 
\addauthor{et}{orange}
\addauthor{aj}{green}
\addauthor{kr}{blue}
\addauthor{fp}{red}
\addauthor{vk}{purple}

\newcommand{\numplayers}{N}

\newcommand{\playeridx}{i}
\newcommand{\playeridxalt}{j}
\DeclarePairedDelimiter\ceil{\lceil}{\rceil}
\newcommand{\abar}{\bar{a}}
\newcommand{\betabar}{\bar{\beta}}
\newcommand{\ubar}[1]{\underaccent{\bar}{#1}}
\newcommand{\R}{\mathbb{R}}
\newcommand{\E}{\mathbb{E}}
\renewcommand{\P}{\mathbb{P}}
\newcommand{\sw}{\mathrm{SW}}
\newcommand{\argmax}[1]{\mathrm{argmax}_{#1} \ }
\newcommand{\lambdamax}[1]{\lambda_{\max}\left(#1\right)}
\newcommand{\lambdamin}[1]{\lambda_{\min}\left(#1\right)}
\newcommand{\spectralradius}[1]{\rho\left(#1\right)}
\newcommand{\inv}[1]{\left(#1\right)^{-1}}
\newcommand{\transpose}{T}
\newcommand{\erdosrenyi}{Erd\H{o}s-R\'enyi }
\newcommand{\lipschitz}{L}
\newcommand{\timeindex}{k}
\newcommand{\constant}{\xi}
\newcommand{\action}{a}
\newcommand{\actionvector}{\mathbf{a}}
\newcommand{\utilityalt}{\tilde{u}}
\newcommand{\utility}{u}
\newcommand{\adjacency}{G}
\newcommand{\potentialfunction}{\phi}

\newcommand{\identity}{I}
\newcommand{\actionvectorpotential}{\actionvector^{\alpha}}
\newcommand{\actionmax}{\bar{\action}}
\newcommand{\actionmin}{\ubar{\action}}
\newcommand{\actiontilde}{\tilde{\action}}
\newcommand{\actionvectoropt}{\actionvector^\mathrm{opt}}
\newcommand{\Gmat}{P}
\newcommand{\Projection}[2]{\mathcal{P}_{#1}\left[#2\right]}
\newcommand{\stepsize}{\eta}
\newcommand{\stepsizemin}{\ubar{\eta}}
\newcommand{\stepsizemax}{\bar{\eta}}
\newcommand{\one}[1]{\mathbf{1}_{#1}}
\newcommand{\actionfixed}{z}
\newcommand{\actionvectorfixed}{\mathbf{\actionfixed}}
\newcommand{\gradientapprox}{g}
\newcommand{\gradientpotential}{h}
\newcommand{\potentialincrement}{\delta}
\newcommand{\gradientbound}{D}
\newcommand{\actionset}{\mathcal{A}}
\newcommand{\localaggregate}{z}
\newcommand{\adelta}{a_{\delta}}
\newcommand{\posboundeigvals}{\mathrm{PoS}_{\gamma, \lambda(\adjacency)}}
\newcommand{\posboundnetworkvals}{\mathrm{PoS}_{\gamma, \adjacency}}
\newcommand{\posboundgamma}{\mathrm{PoS}_\gamma}
\newcommand{\regionone}{\mathcal{R}_{\mathrm{1}}^{k, i}}
\newcommand{\regiontwo}{\mathcal{R}_{\mathrm{2}}^{k, i}}
\newcommand{\regiononea}{\mathcal{R}_{\mathrm{1A}}^{k, i}}
\newcommand{\regiononeb}{\mathcal{R}_{\mathrm{1B}}^{k, i}}
\newcommand{\regiononec}{\mathcal{R}_{\mathrm{1C}}^{k, i}}
\newcommand{\regiontwoa}{\mathcal{R}_{\mathrm{2A}}^{k, i}}
\newcommand{\regiontwob}{\mathcal{R}_{\mathrm{2B}}^{k, i}}
\newcommand{\betavector}{\bm{\beta}}

\begin{document}

\author{Kiran Rokade\inst{1} \and  Adit Jain\inst{1} \and  Francesca Parise\inst{1} \and Vikram Krishnamurthy\inst{1} \and Eva Tardos\inst{2}}

\institute{School of Electrical and Computer Engineering, Cornell University \email{\{kvr36,aj457,fp264,vikramk\}@cornell.edu} \and Department of Computer Science, Cornell University \email{eva.tardos@cornell.edu}}

\title{Asymmetric Network Games: \texorpdfstring{$\alpha$}{alpha}-Potential Function and Learning}

\titlerunning{Asymmetric Network Games: \texorpdfstring{$\alpha$}{alpha}-Potential Function and Learning}

\maketitle

\begin{abstract} 

In a network game, players interact over a network and the utility of each player depends on his own action and on an aggregate of his neighbours' actions. Many real world networks of interest are asymmetric and involve a large number of heterogeneous players. This paper analyzes static network games using the framework of $\alpha$-potential games. Under mild assumptions on the action sets (compact intervals) and the utility functions (twice continuously differentiable) of the players, we derive an expression for an inexact potential function of the game, called the $\alpha$-potential function. Using such a function, we show that modified versions of the sequential best-response algorithm and the simultaneous gradient play algorithm achieve convergence of players' actions to a $2\alpha$-Nash equilibrium. For linear-quadratic network games, we show that $\alpha$ depends on the maximum asymmetry in the network and is well-behaved for a wide range of networks of practical interest. Further, we derive bounds on the social welfare  of the $\alpha$-Nash equilibrium  corresponding to the maximum of the $\alpha$-potential function, under suitable assumptions. We numerically illustrate the convergence of the proposed algorithms and properties of the learned $2\alpha$-Nash equilibria. 


\end{abstract}

\section{Introduction}
\label{sec: introduction}

In a network game, players interacting through a network make strategic decisions aimed at maximizing their utility which depends on their own action and on an aggregate of their neighbors' actions~\cite{jackson_chapter}. 
Much of the existing literature studying learning dynamics in network games relies on global assumptions on the network structure such as symmetric interactions (i.e., $\adjacency = \adjacency^\transpose$, where $\adjacency$ is the adjacency matrix of the network \cite{bramoulle_strategic_interactions}), or spectral conditions which ensure monotonicity of the game Jacobian (e.g., conditions involving $\|\adjacency\|_2$ or $\|\adjacency\|_\infty$ \cite{francesca_VI_network_games}). While these assumptions are useful in establishing properties such as existence and uniqueness of a Nash equilibrium, and convergence of learning algorithms to such an equilibrium, they can be restrictive since (i) many real-world networks exhibit asymmetric interactions, e.g., social networks with unbalanced relationships \cite{asymmetricallycommittedrel} and technological networks with directional dependencies~\cite{Liu_2003_sensor_networks_asymmetric}, and (ii) for large networks with heterogeneous player characteristics, conditions which require a bound on the norm of the network adjacency matrix may not be satisfied. 
Hence, there is a need to study these classes of network games using a framework which is broad (encompassing networks of practical interest that can be large, asymmetric, generated deterministically or randomly) yet enabling tractable analysis of equilibrium outcomes and convergence of independent learning algorithms. 

\textbf{Contributions: }
In this work, we adopt the framework of $\alpha$-potential games to study a broad class of network games not previously covered in the literature. 
Under mild assumptions on the game, we prove existence of a common function (an $\alpha$-potential function) that is approximately optimized\footnote{In Theorems~\ref{thm: sequential best response convergence} and \ref{thm: gradient play convergence} we make the term ``approximately optimized'' precise.} as players selfishly maximize their own utilities, and whose maxima correspond to approximate Nash equilibria.
Our main contributions are as follows. 
\begin{enumerate}
    \item For static games whose utility functions are twice-differentiable, and whose action sets are compact subintervals of $\R$, we provide an explicit form of an $\alpha$-potential function and a bound on the corresponding value of $\alpha$ (Corollary \ref{cor: alpha potential existence}), where $\alpha>0$ is a constant whose value depends on the game parameters.
    \item Using the framework of $\alpha$-potential games, we show convergence of sequential best-response (Algorithm~\ref{alg: sequential best response}) and simultaneous gradient play (Algorithm~\ref{alg: gradient play}) dynamics modified to include a stopping criterion based on $\alpha$ to a $2\alpha$-Nash equilibrium of the game, in a finite number of iterations (Theorems \ref{thm: sequential best response convergence} and \ref{thm: gradient play convergence}). 
    \item For linear-quadratic (LQ) network games, we show that $\alpha$ is proportional to the degree of asymmetry in the network (Proposition~\ref{prop: alpha potential expression LQ}). We then present examples of networks of practical interest for which $\alpha$ is well-behaved (Section \ref{sec: alpha for various networks}). 
    \item Finally, for LQ games, we derive a bound on the welfare suboptimality of the $\alpha$-Nash equilibrium corresponding to the maximum of the provided $\alpha$-potential function in terms of the eigenvalues of the network adjacency matrix (Theorem~\ref{thm: pos}), under suitable assumptions on the parameters of the game. 
\end{enumerate}

\textbf{Related works: }
Our paper adds to the rich literature on \emph{network games} (see, e.g., \cite{jackson_chapter,network_games_galeotti,bramoulle_strategic_interactions,galeotti_interventions,francesca_VI_network_games} among others). 
Within this literature, 
our paper is more closely related to works that study convergence of \emph{learning dynamics}. 
For example, \cite{francesca_VI_network_games} guarantees convergence of discrete- and continuous-time best-response dynamics in network games under 
conditions on the norms of the network adjacency matrix~$\adjacency$. For an LQ game with local aggregate sensitivity parameter $\gamma$,\footnote{See equation \eqref{eq: LQ utility} in Section~\ref{sec: LQ network games} for the definition of an LQ utility function and a precise meaning of $\gamma$.} these conditions reduce to: (i) $|\gamma| \cdot \|\adjacency\|_2 < 1$, (ii) $|\gamma| \cdot \|\adjacency\|_\infty < 1$, (iii) $\gamma < 0$, $|\gamma| \cdot |\lambdamin{\adjacency}| < 1$. 
Although these conditions are broad, i.e., they cover asymmetric network games and games of strategic complements and substitutes, there exist simple network games which do not satisfy them (see Fig. \ref{fig: simple_asymm_networks} for two examples).
\begin{figure}[t]
    \begin{center}
    \begin{align*}
        G_1 = \begin{bmatrix}
                0 &0.95 &0.91 &0.90 \\
                0.97 &0 &0.94 &0.91 \\
                0.93 &0.90 &0 &0.95 \\
                0.95 &0.99  &0.98 &0 
            \end{bmatrix}, \quad \quad
        G_2 = \begin{bmatrix}
                \multicolumn{4}{c}{\multirow{4}{*}{$G_1$}} & 3.02 \\
                 & & & & 2.95 \\
                 & & & & 2.95 \\
                 & & & & 2.96 \\
                3.02 & 2.95 & 2.95 & 2.96 & 0
            \end{bmatrix}
    \end{align*}
    \renewcommand{\arraystretch}{1.2} 
    \setlength{\tabcolsep}{11pt}  
    \begin{tabular}{|c|c|c|c|c|c|} \hline
        &$\adjacency = \adjacency^T$ &$\|\adjacency\|_2$ &$\|\adjacency\|_\infty$ &$\alpha = \|\adjacency - \adjacency^T\|_\infty$ &Max. utility\\ \hline
        $G_1$ &\xmark &$2.82$ &$2.93$ &$0.33$ &$3.30$ \\ \hline
        $G_2$ &\xmark &$7.52$ &$11.91$ &$0.35$ &$12.48$ \\ \hline
    \end{tabular}
    \end{center}
    \caption{Examples of asymmetric networks for which $\|\adjacency\|_2 > 1$ and $\|\adjacency\|_\infty > 1$ (hence not satisfying the conditions in \cite{francesca_VI_network_games} for $\gamma = 1$), while $\alpha$ is small (see Proposition \ref{prop: alpha potential expression LQ}). Max. utility is defined as $\max_{i \in [\numplayers], \actionvector \in [-1,1]^\numplayers} |\utility_\playeridx(\actionvector)|$. $G_1$ can represent a clique of four friends with small errors in the reciprocity of their friendships, while $G_2$ can represent the same network with an additional influential node. 
    }
    \label{fig: simple_asymm_networks}
\end{figure}
Our work applies to these classes of network games. 
Under assumptions that the best-response mapping is either contractive, non-expansive, or monotone, and the network adjacency matrix satisfies $\|\adjacency\|_2 \leq 1$, \cite{PARISE2020Distributed} studies distributed algorithms for convergence to Nash equilibria. Convergence properties of best-response dynamics in network games under conditions such as the adjacency matrix being doubly stochastic and $\|\adjacency\|_2 \leq 1$ are studied in \cite{Grammatico2018Proximal}. 
For time-varying network games sampled from a monotone network game, \cite{feras_learning_LQ_games} provides convergence guarantees for gradient play dynamics. 
Moreover, passivity based approaches for learning in network games have been considered in \cite{Pavel_passivity_networks_TAC}. 
The problem of computing a Nash equilibrium of a graphical game (a general version of a network game) \cite{kearns_2001_graphical_games}, and of public goods games over networks \cite{PAPADIMITRIOU_public_good_directed_network_games} have also been studied in the literature. 
Finally, learning dynamics in special classes of network games have also been considered. 
For example, a particular class of asymmetric network games with a mix of strategic complements and substitutes, called games with coordinating and anti-coordinating agents, is studied in \cite{VANELLI20201Coordinating}. 
A special class of network games, called opinion games, are studied in \cite{BINDEL2015Opinion}. Therein, the authors provide guarantees on the convergence of De-Groot style best-response dynamics and the price of anarchy of the game. 
Network games with discrete action spaces and threshold learning dynamics are studied, for example, in \cite{jackson_behavioral_communities,kempe_diffusion_model}. 
Learning dynamics in \textit{aggregative games}, which are a special class of network games wherein the utility of a player depends on the sum of actions of \emph{all} the players have been considered, for example, in \cite{paccagnan_aggregative_games,ZHU2022Aggregative,Gadjov2021Aggregative}. 
Our work contributes to this literature by ensuring convergence of two learning algorithms to approximate Nash equilibria for a broad class of network games which are neither potential nor monotone, and can have a mix of strategic complements and substitutes. 

Our results are obtained by building on the framework of \emph{near/$\alpha$-potential games}. The concept of near-potential games was introduced in \cite{CANDOGAN_GEB,Candogan_TEAC} for static games. Therein the authors quantified the distance between a game and its nearest potential game using the ``maximum pairwise distance''. The concept of an $\alpha$-potential game was introduced in \cite{guo2025markovalphapotentialgames,guo2025alphapotentialgameframeworknplayer} for (dynamic) Markov games. In this paper, we apply the same notion of $\alpha$-potential games to static games. 
In this setting, any $\alpha$-potential game is a near-potential game as defined in \cite{CANDOGAN_GEB,Candogan_TEAC} with maximum pairwise distance $\alpha$ and vice-versa. 
We provide an explicit form of an $\alpha$-potential function for static games with continuous action spaces by specializing a result from \cite{guo2025alphapotentialgameframeworknplayer} to the static setting. The derived $\alpha$-potential function can be used to construct a potential game that is $\alpha$-close to the original one, yet not necessarily the nearest potential game (as instead studied in \cite{Candogan_TEAC,Candogan_CDC,Matni_CDC_projection_framework}). 
Learning algorithms including best- and better-response dynamics, logit response dynamics and fictitious play for near-potential games have also been considered in the literature (e.g., in \cite{CANDOGAN_GEB,Candogan_CDC,Matni_CDC_projection_framework}). 
We complement these results by analyzing learning algorithms (e.g., approximate gradient dynamics) for $\alpha$-potential games with continuous action spaces and discrete-time updates. Related learning algorithms for Markov $\alpha$-potential games are discussed in \cite{guo2025markovalphapotentialgames}. Finally, a main focus of our work is the application of these results to network games and the characterization of the value of $\alpha$ in the terms of network properties. We note that \cite[Section~6]{guo2025alphapotentialgameframeworknplayer} also applied the framework of $\alpha$-potential games to (dynamic) LQ network games, and discovered the connection of $\alpha$ to network asymmetry. However, the work in \cite{guo2025alphapotentialgameframeworknplayer} focused on deriving an analytical solution of an $\alpha$-Nash equilibrium, while our focus is on independent algorithms for learning $\alpha$-Nash equilibria of static network games.

\textbf{Organization: }
The rest of the paper is organized as follows. In Section \ref{sec: network games and alpha potential function}, we recall the notion of $\alpha$-potential games and study convergence of a sequential best-response algorithm and a simultaneous gradient play algorithm. In Section \ref{sec: LQ network games}, we focus on LQ network games and derive an expression of $\alpha$ for such games, connecting it with the asymmetry in the network. Then, we study the properties of $\alpha$ for various network structures of practical interest. Further, we derive bounds on the welfare suboptimality of the maximizer of an $\alpha$-potential function. In Section \ref{sec: numerical results}, we illustrate our results through numerical simulations. Finally, we give some concluding remarks in Section \ref{sec: conclusion}. Proofs of all the results are given in the appendix.

\textbf{Notation: }
Given a positive integer $\numplayers$, let $[\numplayers] := \{1, 2, \dots, \numplayers\}$. 
Given a vector $\actionvector \in \R^\numplayers$, $a_i$ denotes its $i$th element, while $\actionvector_{-i}$ denotes the vector of all elements of $\actionvector$ except $a_i$. 
$\Projection{A}{b}$ denotes the orthogonal projection of point $b$ onto set the $A$. For a matrix $M \in \R^{\numplayers \times \numplayers}$, $\spectralradius{M}$ denotes its spectral radius. $\mathcal{C}^2(\R^\numplayers)$ denotes the set of twice continuously differentiable functions defined on $\R^\numplayers$.

\section{Learning in \texorpdfstring{$\alpha$}{alpha}-potential games}
\label{sec: network games and alpha potential function}

In this section, we derive a closed-form expression of an $\alpha$-potential function for a general static game. 
We then present two algorithms, a sequential best-response algorithm, and a simultaneous gradient play algorithm with stopping criteria based on $\alpha>0$, whose value depends on the game parameters, and show finite-iteration convergence of the algorithms to the set of $2\alpha$-Nash equilibria of the game.

A static game consists of $\numplayers$ players enumerated using the set $[\numplayers] := \{1, 2, \dots, \numplayers\}$. 
Let the action set of player~$i$ be $\actionset_\playeridx$, and the set of all action profiles be $\actionset := \prod_{\playeridx \in [\numplayers]} \actionset_\playeridx$.
We use $\actionset_{-i} := \prod_{j \neq i} \actionset_j$ to denote the set of all action profiles of players other than player~$i$.  An action $\action_\playeridx \in \actionset_i$ of player $\playeridx$ may represent the level of effort exerted by the player in a group activity \cite{francesca_review}, or the amount of divisible good purchased by a consumer in a society \cite{candogan_optimal_pricing}. 
Given a vector $\actionvector = (\action_1, \dots, \action_\numplayers) \in \actionset$ of actions of all the players, the utility of player $\playeridx$ is given by $\utility_\playeridx(\actionvector) \in \R$.
We make the following assumptions on the action sets and the utility functions of the players. 

\begin{assumption}{(Continuity and compactness)}
\label{assn: continuity and compactness}
    For each $i \in [\numplayers]$: (i) The action set of player $\playeridx$ is a compact interval in $\R$, i.e., $\actionset_\playeridx := [\actionmin_i, \actionmax_i] \subseteq \R$. Let $\actionmax := \sup\{|\action_\playeridx| : \action_\playeridx \in \actionset_\playeridx, i \in [\numplayers]\}$, and $\adelta := \sup\{|\action_\playeridx - \action_\playeridx'| :  \action_\playeridx, \action_\playeridx' \in \actionset_\playeridx, i \in [\numplayers]\}$. (ii) The utility function $\utility_\playeridx \in \mathcal{C}^2(\R^\numplayers)$ and $\max_{\actionvector \in \actionset} \left|\partial/(\partial a_i) \utility_\playeridx(\actionvector)\right| \leq \gradientbound$ for some $\gradientbound > 0$. 
\end{assumption}

Next, we recall the definition of an $\epsilon$-Nash equilibrium of a game. 

\begin{definition}{($\epsilon$-Nash equilibrium)}
\label{def: epsilon nash equilibrium}
    An action vector $\actionvector = (\action_1, \dots, \action_\numplayers) \in \actionset$ is said to be an  $\epsilon$-Nash equilibrium of the game if for all players $\playeridx \in [\numplayers]$ and all actions $\action^\prime_\playeridx \in \actionset_i$,  
    \begin{align*}
        \utility_\playeridx(\actionvector) \geq  \utility_\playeridx(\action^\prime_\playeridx, \actionvector_{-\playeridx}) - \epsilon. 
    \end{align*}
    A $0$-Nash equilibrium is a Nash equilibrium of the game. 
\end{definition}

\subsection{\texorpdfstring{$\alpha$-potential function}{alpha-potential function}}

Our primary tool to characterize approximate Nash equilibria and to study learning algorithms is the concept of an $\alpha$-potential function (see, e.g., \cite{guo2025markovalphapotentialgames} for dynamic Markov games). 

\begin{definition}{($\alpha$-potential function)}
\label{def: alpha potential function}
    A function $\potentialfunction: \actionset \rightarrow \R$ is said to be an $\alpha$-potential function of the game if for all players $\playeridx \in [\numplayers]$ and all actions $\action_\playeridx, \action_\playeridx^\prime \in \actionset_i, \actionvector_{-\playeridx} \in \actionset_{-i}$, it holds that
    \begin{align}
    \label{eq: alpha potential condition}
        \left\vert\left[\utility_\playeridx(\actionvector) - \utility_\playeridx(\action_\playeridx^\prime,\actionvector_{-\playeridx})\right] - \left[\potentialfunction(\actionvector) - \potentialfunction(\action_\playeridx^\prime,\actionvector_{-\playeridx})\right]\right\vert \leq \alpha. 
    \end{align}
    A $0$-potential function of the game is an exact potential function or simply a potential function of the game \cite{shapley_potential}. 
\end{definition}

\begin{remark}
    The parameter $\alpha$ can be intuitively understood as a measure of how far a game deviates from being an exact potential game (see also \cite{CANDOGAN_GEB,Candogan_TEAC}). In exact potential games (in which $\alpha$ = 0), players' incentives are perfectly aligned with a potential function, ensuring that maxima of the potential function are Nash equilibria of the game \cite{shapley_potential}. 
    As $\alpha$ increases, this alignment weakens. 
\end{remark}

The following result provides an explicit form of an $\alpha$-potential function for a static game, and can be obtained by specializing \cite[Theorem~2.5]{guo2025alphapotentialgameframeworknplayer}, which was formulated for dynamic games, to static games. 

\begin{corollary}{(Characterization of an $\alpha$-potential function)}
\label{cor: alpha potential existence}
    Consider a game satisfying Assumption~\ref{assn: continuity and compactness} and the function $\potentialfunction : \actionset \to \R$ specified as 
    \begin{align}
    \label{eq: alpha potential expression}
        \potentialfunction(\actionvector)=\int_{0}^{1} \sum_{\playeridxalt=1}^{\numplayers}  \frac{\partial \utility_\playeridxalt}{\partial \action_\playeridxalt}(\actionvectorfixed + r(\actionvector-\actionvectorfixed)) \cdot (\action_\playeridxalt - \actionfixed_\playeridxalt) \ dr 
    \end{align}
    where $\actionvectorfixed \in \actionset$ is any fixed action vector. 
    Then $\potentialfunction$ is an $\alpha$-potential function of the game (Definition~\ref{def: alpha potential function}) 
    with $\alpha$ given by 
    \begin{align}
        \label{eq: alpha bound} 
        \alpha = \frac{\adelta^2}{2} \max_{\playeridx \in [\numplayers], \actionvector\in\actionset} \sum_{\playeridxalt=1}^\numplayers \left|\frac{\partial^2 \utility_\playeridx}{\partial \action_\playeridx \partial \action_\playeridxalt}(\actionvector) - \frac{\partial^2 \utility_\playeridxalt}{\partial \action_\playeridxalt \partial \action_\playeridx}(\actionvector)\right|. 
    \end{align}
    Further, $\potentialfunction$ is differentiable and for all $\actionvector \in \actionset$, 
    \begin{align}
    \label{eq: phi and u_i derivative relation}
        \left|\frac{\partial \potentialfunction}{\partial \action_\playeridx}(\actionvector) - \frac{\partial \utility_\playeridx}{\partial \action_\playeridx}(\actionvector)\right| \leq \frac{\alpha}{\adelta}. 
    \end{align}
\end{corollary}

The proof of Corollary \ref{cor: alpha potential existence} follows similar steps as that of \cite[Theorem~2.5]{guo2025alphapotentialgameframeworknplayer} on Markov games. For completeness, we provide a proof in Appendix~\ref{sec: proof of theorem alpha potential existence} by specializing the arguments to static games.  
In Section \ref{sec: LQ network games}, we provide a simplified expression of $\alpha$ for LQ network games, and relate it to the asymmetry of the network adjacency matrix. 

Next, we present two algorithms using which players can learn a $2\alpha$-Nash equilibrium. By learning we mean that the players update their actions so that the action profile converges to a $2\alpha$-Nash equilibrium.

\subsection{Sequential best-response algorithm}

The first algorithm we propose, Algorithm~\ref{alg: sequential best response}, is a slight modification of the well-known sequential best-response algorithm: At each time step, at most one player can update his action by playing his best response to the current set of actions of all the other players, \emph{provided he obtains a utility that is larger than his current utility by $\alpha + \epsilon$}, where $\epsilon > 0$. The update step is motivated by the following assumption on players' behaviour. 

\vspace{1ex}
\noindent \textit{Stopping criterion rule.} \textit{For all players $\playeridx\in[\numplayers]$, player $\playeridx$ will not deviate from an action vector $\actionvector \in \actionset$ if
$\max_{\action^\prime_\playeridx \in \actionset_i }\utility_\playeridx(\action^\prime_\playeridx,\actionvector_{-i})-\utility_\playeridx(\action_\playeridx,\actionvector_{-i}) \leq \alpha + \epsilon$, i.e., if the improvement in his utility is at most $\alpha + \epsilon$.}
\vspace{1ex}

\begin{algorithm}[t]
    \begin{algorithmic}[1]
        \STATE Fix $\alpha > 0$, $\epsilon > 0$. For each $i \in [\numplayers]$, initialize $a_i^0 \in \actionset_i$ arbitrarily.  
        \FOR{$\timeindex = 0, 1, 2,  \dots$}
        \IF{$\exists i \in [\numplayers]$ such that $\max_{a_i \in \actionset_i} \utility_\playeridx(a_i, \actionvector_{-i}^k) > \utility_\playeridx(\actionvector^k) + \alpha + \epsilon$}
        \STATE Pick any such $i \in [\numplayers]$. 
        \STATE Update $\action_{\playeridx}^{\timeindex+1} = \argmax{\action_\playeridx \in \actionset_i}\utility_\playeridx(\action_\playeridx,\actionvector_{-\playeridx}^{\timeindex})$, and fix $a_j^{k+1} = a_j^k$ for all $j \neq i$. 
        \ELSE 
        \STATE $\actionvector^{k+1} = \actionvector^k$. 
        \ENDIF
        \ENDFOR
    \end{algorithmic}
    \caption{Sequential best-response conditioned on utility improvement}
    \label{alg: sequential best response}
\end{algorithm}

Our first main result guarantees convergence of the action profiles learned via Algorithm~\ref{alg: sequential best response} to an $(\alpha + \epsilon)$-Nash equilibrium in finitely many iterations for any game with an $\alpha$-potential function. 

\begin{theorem}{(Convergence of Algorithm~\ref{alg: sequential best response})}
\label{thm: sequential best response convergence}
    Consider a network game satisfying Assumption~\ref{assn: continuity and compactness} with an $\alpha$-potential function $\potentialfunction$ (Definition~\ref{def: alpha potential function}) for some $\alpha > 0$. Let $\potentialfunction^{\max} := \max_{\actionvector \in \actionset} |\potentialfunction(\actionvector)| < \infty$, and $\bar{k} := \lceil\frac{2\potentialfunction^{\max}}{\epsilon} \rceil$. Let $\{\actionvector^k\}_{k = 0}^\infty$ be a sequence of action profiles generated by Algorithm~\ref{alg: sequential best response} for any fixed $\epsilon > 0$. Then, for all $k \geq \bar{k}$, $\actionvector^k = \actionvector^*$ where $\actionvector^*$ is an $(\alpha + \epsilon)$-Nash equilibrium of the game (Definition~\ref{def: epsilon nash equilibrium}). 
\end{theorem}

Note that Theorem~\ref{thm: sequential best response convergence} holds with the $\alpha$-potential function given in~\eqref{eq: alpha potential expression} and $\alpha$ in~\eqref{eq: alpha bound}. 
The proof of Theorem~\ref{thm: sequential best response convergence} given in Appendix~\ref{sec: proof of theorem sequential best-response convergence} is a simple extension of the idea of the ``approximate finite improvement property'' \cite{shapley_potential}: by definition of an $\alpha$-potential function, we show that at every step of Algorithm~\ref{alg: sequential best response}, improving the utility of a player by $\alpha + \epsilon$ increases the $\alpha$-potential function by $\epsilon > 0$. 

\begin{remark}
    Results similar to Theorem~\ref{thm: sequential best response convergence} have been proven earlier for the case of finite action spaces for both discrete-time \cite[Theorem~3.1]{CANDOGAN_GEB} and continuous-time best-response dynamics \cite[Theorem~5.7]{Candogan_TEAC}, as well as for games with continuous action spaces and polynomial utility functions \cite[Proposition~2]{Matni_CDC_projection_framework}. We state our result for the case of continuous action spaces, discrete-time best-response dynamics and any set of utility functions satisfying Assumption \ref{assn: continuity and compactness}.  
\end{remark}

Note that Theorem~\ref{thm: sequential best response convergence} only ensures convergence to an $(\alpha + \epsilon)$-Nash equilibrium of the game. 
However, the result applies to a very broad class of network games. Specifically, it does not require concavity of the utility functions or monotonicity of the game as in~\cite{francesca_VI_network_games}, or symmetry of the network or non-negativity of $G_{ij}$'s as in~\cite{bramoulle_strategic_interactions}, which are standard assumptions in the network games literature.

\subsection{Simultaneous gradient play algorithm}
\label{sec: gradient play algorithm}

Our next algorithm, Algorithm~\ref{alg: gradient play}, is a variant of the simultaneous gradient play algorithm. In this algorithm, at each iteration, players simultaneously update their actions based on feedback on the gradient of their utility function, \emph{provided their gradient is large enough}. 

\begin{algorithm}[t]
    \begin{algorithmic}[1]
        \STATE Fix $\alpha > 0$. Let $\adelta := \sup\{|\action_\playeridx - \action_\playeridx'| :  \action_\playeridx, \action_\playeridx' \in \actionset_\playeridx, i \in [\numplayers]\}$. For each $i \in [\numplayers]$, fix $\stepsize_i > 0$ and initialize $\action_\playeridx^0 \in \actionset_i$ arbitrarily.  
        \FOR{$\timeindex = 0, 1, 2, \dots$}
        \FOR{$\playeridx \in [\numplayers]$ }
        \IF{$\left|\partial/(\partial a_i)\utility_\playeridx(\actionvector^{\timeindex})\right| > 2\alpha/\adelta$}
        \STATE Update $\action_{\playeridx}^{\timeindex+1} = \Projection{\actionset_i}{\action_{\playeridx}^{\timeindex} + \stepsize_\playeridx \partial/(\partial a_i)\utility_\playeridx(\actionvector^{\timeindex})}$. 
        \ELSE
        \STATE $\action_{\playeridx}^{\timeindex+1} = \action_{\playeridx}^{\timeindex}$. 
        \ENDIF
        \ENDFOR
        \ENDFOR
    \end{algorithmic}
    \caption{Simultaneous gradient play conditioned on gradient value}
    \label{alg: gradient play}
\end{algorithm}

Our next main result guarantees that players learn a $2\alpha$-Nash equilibrium of the game using Algorithm~\ref{alg: gradient play}, provided the step sizes used by the players are small enough. 
We will need the following assumption on the utility functions of the game to ensure that any local $\epsilon$-Nash equilibrium is also a global $\epsilon$-Nash equilibrium of the game. 

\begin{assumption}{(Concavity)}
\label{assn: utility functions concave}
    For all $i \in [\numplayers]$, the utility function $\utility_\playeridx(a_i, \actionvector_{-i})$ is concave in $a_i$ for every fixed $\actionvector_{-i} \in \actionset_{-i}$. 
\end{assumption}

\begin{theorem}{(Convergence of Algorithm~\ref{alg: gradient play})}
\label{thm: gradient play convergence}
     Consider a game satisfying Assumptions~\ref{assn: continuity and compactness} and \ref{assn: utility functions concave}. Let $\potentialfunction$ be an $L$-smooth $\alpha$-potential function of the game with $\alpha > 0$, $L > 0$ (see Definition~\ref{def: alpha potential function}) that is bounded from above and
     satisfies the inequality \eqref{eq: phi and u_i derivative relation}. Let $\stepsizemin := \min_{i \in [\numplayers]} \stepsize_i > 0$, $\stepsizemax := \max_{i \in [\numplayers]} \stepsize_i$, $\constant_1 := 1/\adelta$ and $\constant_2 := \min\{2 \constant_1/\lipschitz,  1/\gradientbound\}$. Suppose $\stepsizemax \leq \min\left\{1/\lipschitz, \ 2\constant_1^2\alpha^2/(\lipschitz\gradientbound^2), \ \lipschitz \constant_2^2 \alpha^2 / (2 (\constant_1 \alpha \gradientbound + \constant_1^2 \alpha^2))\right\}$. Let $\{\actionvector^k\}_{k = 0}^\infty$ be the sequence of action profiles generated by Algorithm~\ref{alg: gradient play}. 
     Then, $\exists \bar{k} \geq 0$ such that for all $k \geq \bar{k}$, $\actionvector^k$ is a $2\alpha$-Nash equilibrium of the game (Definition~\ref{def: epsilon nash equilibrium}). 
\end{theorem}

Note that Theorem~\ref{thm: gradient play convergence} holds with the $\alpha$-potential function given in~\eqref{eq: alpha potential expression} and $\alpha$ in~\eqref{eq: alpha bound}  provided $\utility_\playeridx \in \mathcal{C}^3(\R^\numplayers)$ for all players $i$ (see Lemma \ref{lem: alpha potential properties} in Appendix \ref{sec: appendix intermediate results}). 
The proof of Theorem~\ref{thm: gradient play convergence} given in Appendix~\ref{sec: proof of theorem simultenous gradient play} proceeds by showing that, at every iteration $k$, either the value of the $\alpha$-potential function increases by a positive amount independent of $k$, or the trajectory is in a region of the action sets in which all action profiles are $2\alpha$-Nash equilibria. Since the $\alpha$-potential function is bounded from above, it follows that the trajectory must be a $2\alpha$-Nash equilibrium after at most finitely many iterations. 

\begin{remark}
    Exact gradient dynamics have been extensively studied for (exact) potential and monotone games (see, e.g., \cite{Mazumdar_2020},\cite{facchinei2003finite} among others), ensuring convergence to local or global Nash equilibria. 
    Theorem~\ref{thm: gradient play convergence} instead shows convergence of modified gradient dynamics to the set of $2\alpha$-Nash equilibria for games with an $\alpha$-potential function. 
\end{remark}

Next, we study a special class of network games called linear-quadratic (LQ) network games which help us connect the value of $\alpha$ in the $\alpha$-potential function of a network game to the asymmetry in the network structure. 
Note that the algorithms above can be directly applied to any network game satisfying Assumption \ref{assn: continuity and compactness} (and \ref{assn: utility functions concave}).

\section{Linear-quadratic (LQ) network games}
\label{sec: LQ network games}

In a network game, we assume that players are part of a network so that the utility obtained by a player depends on his own action and on the actions of his neighbors in the network. In particular, let $\adjacency \in \R^{\numplayers \times \numplayers}$ be the weighted adjacency matrix of the network such that $\adjacency_{\playeridx\playeridxalt} \neq 0$ if the utility of player $\playeridx$ depends on player $\playeridxalt$'s action, and $\adjacency_{\playeridx\playeridxalt} = 0$ otherwise.
We use the convention that $G_{ij} > 0$ if and only if there exists an edge from node $j$ to node $i$ in the underlying network. 
We fix $\adjacency_{\playeridx\playeridx} = 0$ for all $\playeridx \in [\numplayers]$. Let $\localaggregate_\playeridx(\actionvector_{-i}) := \sum_{\playeridxalt = 1}^\numplayers \adjacency_{\playeridx\playeridxalt}\action_\playeridxalt$ denote the local aggregate of player $\playeridx$. 
Given a vector $\actionvector \in \actionset$ of actions of all the players, the utility of player $\playeridx$ is given by the function $\utilityalt_\playeridx(\action_\playeridx,\localaggregate_\playeridx(\actionvector_{-i}))$. For ease of notation, we define $\utility_\playeridx(\actionvector) = \utilityalt_\playeridx(\action_\playeridx,\localaggregate_\playeridx(\actionvector_{-i}))$.\footnote{Note that $\utilityalt_i : \R^2 \to \R$, while $\utility_\playeridx : \R^\numplayers \to \R$.}

In the following, we focus on LQ network games. In an LQ network game, the utility function for each player $i \in [\numplayers]$ is given by
\begin{align}
    \label{eq: LQ utility}\utility_\playeridx(\actionvector) = -\frac{1}{2} \action_\playeridx^2 + \beta_i \action_\playeridx + \gamma \localaggregate_i(\actionvector_{-i}) \action_\playeridx, 
\end{align}
where $\beta_i \in \R$ is player~$i$'s bias towards his own action and $\gamma \in \R$ is the player's affinity towards the action of his neighbours in the network.\footnote{Even though we assume that the parameter $\gamma$ is the same for all players, each player $i \in [\numplayers]$ may have a different set of values $\{G_{ij} : j \in [\numplayers]\}$. In that sense, the utility function in \eqref{eq: LQ utility} is equivalent to a general LQ utility function.} These games are interesting since (i) they can model various scenarios such as a opinion formation in a social network \cite{BINDEL2015Opinion}, purchasing certain quantities of a good in a social network \cite{candogan_optimal_pricing}, choosing the amount of effort in a group activity \cite{jackson_chapter}, etc., and (ii) the LQ structure in the utility function helps in computations (e.g., leading to a closed-form expression of a Nash equilibrium in~\cite{Ballester_2006}). In this section, we exploit this structure to provide a direct connection between the value of $\alpha$ and the degree of asymmetry in the network adjacency matrix~$\adjacency$.

\subsection{\texorpdfstring{$\alpha$-potential function}{alpha-potential function}}

Given an LQ network game, consider the function $\potentialfunction: \actionset \rightarrow \R$ such that
\begin{align}
\label{eq: LQ alpha potential function}
    \potentialfunction(a) := -\frac{1}{2} \sum_{i = 1}^\numplayers a_i^2 + \sum_{i = 1}^\numplayers \beta_i a_i + \frac{\gamma}{2} \sum_{i = 1}^\numplayers \left(\sum_{j = 1}^\numplayers G_{ij} a_j\right) a_i.  
\end{align}
It is known that $\potentialfunction$ is an exact potential function for the network game if the network is symmetric i.e., $\adjacency = \adjacency^T$ (see \cite{bramoulle_strategic_interactions}). We show that $\potentialfunction$ is an $\alpha$-potential function of the game if the network is asymmetric, and relate $\alpha$ to the degree of asymmetry in the network. 

\begin{proposition}
\label{prop: alpha potential expression LQ}
    For an LQ network game satisfying Assumption~\ref{assn: continuity and compactness}, the function $\potentialfunction$ defined in \eqref{eq: LQ alpha potential function} is an $\alpha$-potential function of the game with 
    \begin{align}
    \label{eq: alpha bound LQ game}
        \alpha = \frac{\abar \adelta |\gamma|}{2} \|\adjacency - \adjacency^T\|_\infty.  
    \end{align} 
\end{proposition}

A proof of Proposition~\ref{prop: alpha potential expression LQ} is given in Appendix~\ref{sec: proof of alpha potential expression LQ}. It proceeds by verifying the definition of an $\alpha$-potential function (Definition~\ref{def: alpha potential function}) for the function in \eqref{eq: LQ alpha potential function}.\footnote{Note that the expression of $\potentialfunction$ in \eqref{eq: LQ alpha potential function} can be obtained by substituting $\actionvectorfixed = 0$ in \eqref{eq: alpha potential expression}. However, we state Proposition~\ref{prop: alpha potential expression LQ} independently since $0 \notin \actionset$ in general, and the bound in \eqref{eq: alpha bound LQ game} is different from the one in \eqref{eq: alpha bound}. }
The bound on $\alpha$ in \eqref{eq: alpha bound LQ game} is a constant multiple of $\|\adjacency - \adjacency^T\|_\infty = \max_{i \in [\numplayers]} \sum_{\playeridxalt = 1}^\numplayers \left|\adjacency_{\playeridx\playeridxalt} - \adjacency_{\playeridxalt\playeridx}\right|$. Thus, if $G_{ij} \in \{0, 1\}$, then $\alpha$ is proportional to the maximum number of asymmetric links across all players in the network, and hence is a measure of the degree of asymmetry in the network. In the next section, we present some networks for which the asymmetry quantified by $\|\adjacency - \adjacency^T\|_\infty$, and thus the value of $\alpha$, is small.

\begin{remark}
    The bound given in \eqref{eq: alpha bound LQ game} is similar to the one in \cite[Equation (6.4)]{guo2025alphapotentialgameframeworknplayer}. 
    However, \cite{guo2025alphapotentialgameframeworknplayer} considers a continuous-time dynamic linear-quadratic network game. 
\end{remark}

\subsection{\texorpdfstring{$\alpha$}{alpha} for various networks}
\label{sec: alpha for various networks}

We next present examples of network games for which the value of $\alpha$ given by \eqref{eq: alpha bound LQ game} is well-behaved. For reference, we compare our results with \cite{francesca_VI_network_games}, which provides conditions for convergence of (exact) best-response dynamics in asymmetric network games. 

Consider an LQ network game with a utility function as in \eqref{eq: LQ utility}, $\gamma = 1$ and let the action sets be $\actionset_i = [-1,1], \forall \playeridx \in [\numplayers]$.
For this game, Theorems~\ref{thm: sequential best response convergence} and \ref{thm: gradient play convergence} guarantee convergence of the sequence of learned actions to the set of $2\alpha$-Nash equilibria where $\alpha = \|\adjacency - \adjacency^T\|_\infty$. On the other hand, \cite{francesca_VI_network_games} guarantees convergence of best-response dynamics to a Nash equilibrium if either $\|\adjacency\|_2 < 1$ or $\|\adjacency\|_\infty < 1$.\footnote{See \cite[Theorems~1~and~2, Assumption~2 and Table~2]{francesca_VI_network_games}. Note that the condition on $\lambdamin{\adjacency}$ in \cite[Assumption~2]{francesca_VI_network_games} applies  to symmetric network games. Hence, we do not consider it here.} Next, we present examples of network models which capture a variety of practical scenarios and have a ``reasonable'' value of $\alpha$ even as the number of players increases, while they do not usually satisfy the conditions $\|\adjacency\|_2 < 1$, $\|\adjacency\|_\infty < 1$. The results are summarized in Table \ref{tab: alpha values} and Fig. \ref{fig: alpha for various networks}. 

\begin{table}[t]
    \centering
    \renewcommand{\arraystretch}{1.2} 
    \setlength{\tabcolsep}{5pt}
    \begin{tabular}{|c|c|c|} \hline
        \textbf{Network model} & \textbf{Parameters} & \textbf{$\alpha = \|\adjacency - \adjacency^T\|_\infty$} \\ \hline
        Complete network with errors & $\epsilon(\numplayers) = 1/\numplayers^r$, $r > 1$ & $2/\numplayers^{r-1}$ \\ \hline
        Complete network with an influential player & $\epsilon(\numplayers) = 1/\numplayers^r$, $r > 1$, $w > 1$ & $2/\numplayers^{r-1}$ \\ \hline 
        Complete network with random signs & $\epsilon(\numplayers) = \delta(\numplayers) = 1/\numplayers^r$, $r > 2$ & $4/\numplayers^{r-1} + t$, w.h.p. \\ \hline
        \erdosrenyi network & $p(\numplayers) = 1/\numplayers^r$ or $p(\numplayers) = 1 - 1/\numplayers^r$, $r > 2$ & $2/\numplayers^{r-1} + t$, w.h.p. \\ \hline 
        Small world network &  $d(\numplayers) = \numplayers / 8$, $p(\numplayers) = 1/\numplayers^r$, $r > 2$ & $1/\numplayers^{r-1} + t$, w.h.p. \\ \hline 
        Star network with erased edges & $p(\numplayers) = 1/\numplayers^r$, $r > 2$ & $2/\numplayers^{r-1} + t$, w.h.p. \\ \hline 
    \end{tabular}
    \caption{A summary of bounds on the value of $\alpha$ for various network models and the LQ game described in Section \ref{sec: alpha for various networks}. The parameter $t > 0$ can be arbitrarily small, and ``w.h.p.'' stands for ``with high probability''.}
    \label{tab: alpha values}
\end{table}

\begin{example}{(Complete network with errors)}
\label{ex: complete network with errors}
    Consider a network of $\numplayers$ players where, for all $i \neq j$, $G_{ij} \in [1-\epsilon(\numplayers), 1 + \epsilon(\numplayers)]$ and $\epsilon(\numplayers)$ is a small positive number. 
    One can think of $\adjacency$ as a network of relationships among a clique of friends with an error of at most $2 \epsilon(\numplayers)$ in the reciprocity of their relationships. We show in Appendix~\ref{sec: proofs of alpha and G norm bounds} that, for this network, $\alpha = \|\adjacency - \adjacency^T\|_\infty \leq 2 \numplayers \epsilon(\numplayers)$, while $\|\adjacency\|_2 \geq \sqrt{\numplayers (\numplayers-2)} (1 - \epsilon(\numplayers))$ and $\|\adjacency\|_\infty \geq (\numplayers-1) (1 - \epsilon(\numplayers))$. 
    Suppose $\numplayers \epsilon(\numplayers) \to 0$ as $\numplayers \to \infty$, e.g., $\epsilon(\numplayers) = 1/\numplayers^2$. In this setting, as $\numplayers \to \infty$, $\alpha  = \|\adjacency - \adjacency^T\|_\infty \to 0$, whereas $\min\{\|\adjacency\|_2, \|\adjacency\|_\infty\} \geq \sqrt{\numplayers(\numplayers-2)} (1 - \epsilon(\numplayers)) \to \infty$ as $\numplayers \to \infty$. Thus, for a large $\numplayers$, the convergence results in \cite{francesca_VI_network_games} cannot be applied, while Algorithms~\ref{alg: sequential best response} and \ref{alg: gradient play} ensure convergence to the set of approximate Nash equilibria of the game, with approximation error ($2\alpha$) going to zero as $\numplayers\to\infty$. 
\end{example}

\begin{example}{(Complete network with an influential player)} 
    Consider a complete network with errors as described in Example \ref{ex: complete network with errors}, but suppose the edge weights of player $\numplayers$ are modified such that $G_{Nj}, G_{jN} \in [w -\epsilon(\numplayers), w + \epsilon(\numplayers)]$ for all $j \in [\numplayers]$ where $w$ is the ``average'' weight of a link formed by player $\numplayers$. If $w \gg 1$, we can think of player $\numplayers$ as a highly influential player in the network. Using similar techniques as in Example \ref{ex: complete network with errors}, we show in Appendix~\ref{sec: proofs of alpha and G norm bounds} that $\alpha = \|\adjacency - \adjacency^T\|_\infty \leq 2 \numplayers \epsilon(\numplayers)$, while $\|\adjacency\|_2 \geq \sqrt{\numplayers(\numplayers-2)} (1 - \epsilon(\numplayers))$ for any $w \geq 1$ and $\|\adjacency\|_\infty \geq (\numplayers-1) (w - \epsilon(\numplayers))$. In particular, the value of $w$ does not affect $\alpha = \|\adjacency - \adjacency^T\|_\infty$. 
\end{example}

\begin{example}{(Complete network with random signs)}
\label{ex: signed complete network}
    Consider a signed version of the complete network with errors as described in Example \ref{ex: complete network with errors}.
    Fix two small positive numbers $\epsilon(\numplayers)$ and $\delta(\numplayers)$ that represent the errors in the reciprocity and the signs of the relationships, respectively. For each pair of nodes $i \neq j$, with probability $1/2$, both $G_{ij}, G_{ji}$ are sampled independently and uniformly at random from $[1 - \epsilon(\numplayers), 1 + \epsilon(\numplayers)]$, and with probability $1/2$, both $G_{ij}, G_{ji}$ are sampled independently and uniformly at random from $[-1 - \epsilon(\numplayers), -1 + \epsilon(\numplayers)]$. 
    To introduce differences in the signs among the relationships of the neighboring nodes, for each $i < j$, set $G_{ij} = -G_{ji}$ with probability $\delta(\numplayers)$, and keep $G_{ij}, G_{ji}$ unchanged with probability $1 - \delta(\numplayers)$. It can be shown that, if $\epsilon(\numplayers) = \delta(\numplayers) = 1/\numplayers^r$ with $r > 2$, then, for any fixed $\delta > 0$ and $t > 0$, for $\numplayers$ large enough, with probability $1 - \delta$, $\alpha = \|\adjacency - \adjacency^T\|_\infty \leq 4/\numplayers^{r-1} + t$. See Appendix~\ref{sec: proofs of alpha and G norm bounds} for a proof. 
\end{example}

\begin{example}{(\erdosrenyi network)}
\label{ex: dense erdos renyi networks}
    Consider an \erdosrenyi model for generating directed random networks as follows. Given $\numplayers$ nodes, for each pair of nodes $i \neq j$, we set $G_{ij} = 1$ with probability $p(\numplayers)$, and $G_{ij} = 0$ otherwise. Suppose that, for a fixed $r > 2$, either $p(\numplayers) \leq 1/\numplayers^r$ (sparse network) or $p(\numplayers) \geq 1 - 1/\numplayers^r$ (dense network). Then, it can be shown that, for any $\delta > 0$ and $t > 0$, for $\numplayers$ large enough, with probability $1 - \delta$, $\alpha = \|\adjacency - \adjacency^T\|_\infty \leq 2 / \numplayers^{r-1} + t$. See Appendix~\ref{sec: proofs of alpha and G norm bounds} for a proof. 
\end{example}

\begin{example}{(Small world network)}
\label{ex: small world networks}
    Consider a directed version of the small world network model introduced in \cite{watts1998small-world}. Specifically, given $\numplayers$ nodes, we first arrange all nodes in the form of a ring. Then, we draw a directed edge from each node $i$ to each of its $d(\numplayers)$ adjacent nodes on either side. This gives rise to a regular ring network, each node having an out-degree and in-degree $2 d(\numplayers)$. Second, we rewire each edge with probability $p(\numplayers)$ by disconnecting its head and re-connecting it to one of the $\numplayers - 1$ nodes uniformly at random. We set $G_{ij} = 1$ if there exists at least one edge from node $j$ to node $i$ in the rewired network, and $G_{ij} = 0$ otherwise. 
    Suppose $d(\numplayers) = \numplayers / 8$, $p(\numplayers) = 1/\numplayers^r$ with $r > 2$. Then, it can be shown that, for any $\delta > 0$ and $t > 0$, for $\numplayers$ large enough,  with probability $1 - \delta$, $\alpha = \|\adjacency - \adjacency^T\|_\infty \leq 1/\numplayers^{r-1} + t$. See Appendix~\ref{sec: proofs of alpha and G norm bounds} for a proof. 
\end{example}

\begin{example}{(Star network with erased edges)}
\label{ex: star networks}
    Consider a star network of $\numplayers$ nodes where each node $i > 1$ has a directed edge to node 1 $(G_{1i} = 1)$ and vice-versa $(G_{i1} = 1)$, while no other edges are present in the network. To introduce asymmetry, we delete each edge with probability $p(\numplayers)$. Suppose $p(\numplayers) = 1/\numplayers^r$ with $r > 2$. Then, it can be shown that, for any $\delta > 0$ and $t > 0$, for $\numplayers$ large enough,  with probability $1 - \delta$, $\alpha = \|\adjacency - \adjacency^T\|_\infty \leq 2/\numplayers^{r-1} + t$. See Appendix~\ref{sec: proofs of alpha and G norm bounds} for a proof. 
\end{example}

For illustration, a plot of $\alpha = \|\adjacency - \adjacency^T\|_\infty$, $\|\adjacency\|_2$ and $\|\adjacency\|_\infty$ versus the size of the network ($\numplayers$) for the network models described above is shown in Figure \ref{fig: alpha for various networks}, averaged over $100$ random realizations of each network. It is observed that the value of $\alpha$ scales well with $\numplayers$, while the conditions $\|\adjacency\|_2 < 1$ and $\|\adjacency\|_\infty < 1$ required to show convergence of best-response dynamics in \cite{francesca_VI_network_games} are violated. Thus, the framework of $\alpha$-potential games provides a tractable way to analyze such network games not covered in the literature. 
In order to truly quantify the ``quality'' of the approximate Nash equilibria, one needs to study the value of $\alpha$ relative to the players' utilities. We defer this to Section \ref{sec: numerical results}. 
\begin{figure}[t]
    \centering
    \begin{subfigure}[t]{0.25\textwidth}
        \centering
        \includegraphics[width=\linewidth]{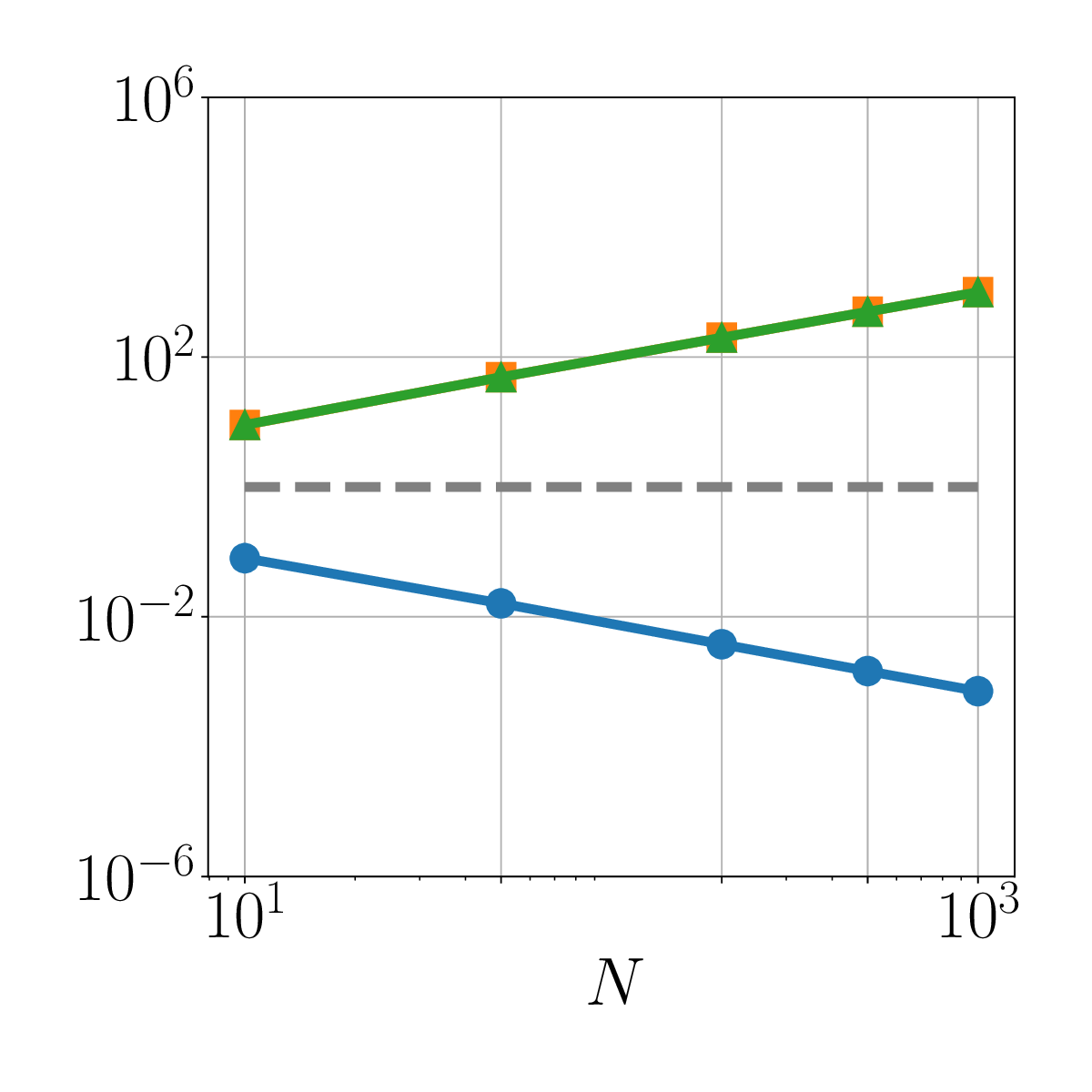}
        \caption{Complete networks with errors, $\epsilon(\numplayers) = 1/\numplayers^2$}
    \end{subfigure}%
    \hfill
    \begin{subfigure}[t]{0.25\textwidth}
        \centering
        \includegraphics[width=\linewidth]{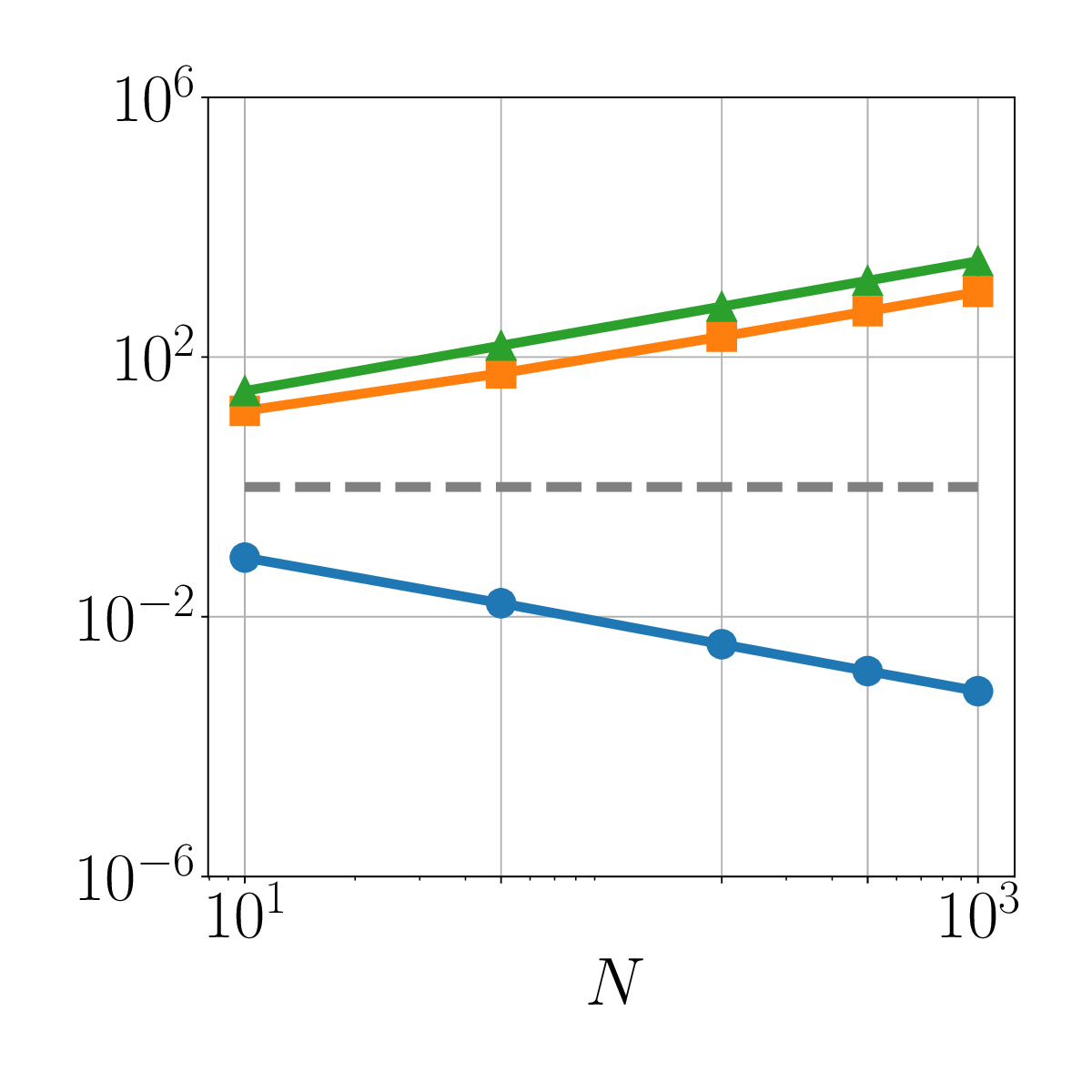}
        \caption{Complete network with an influential node, $\epsilon(\numplayers) = 1/\numplayers^2$, $w = 3$}
    \end{subfigure}%
    \hfill
    \begin{subfigure}[t]{0.25\textwidth}
        \centering
        \includegraphics[width=\linewidth]{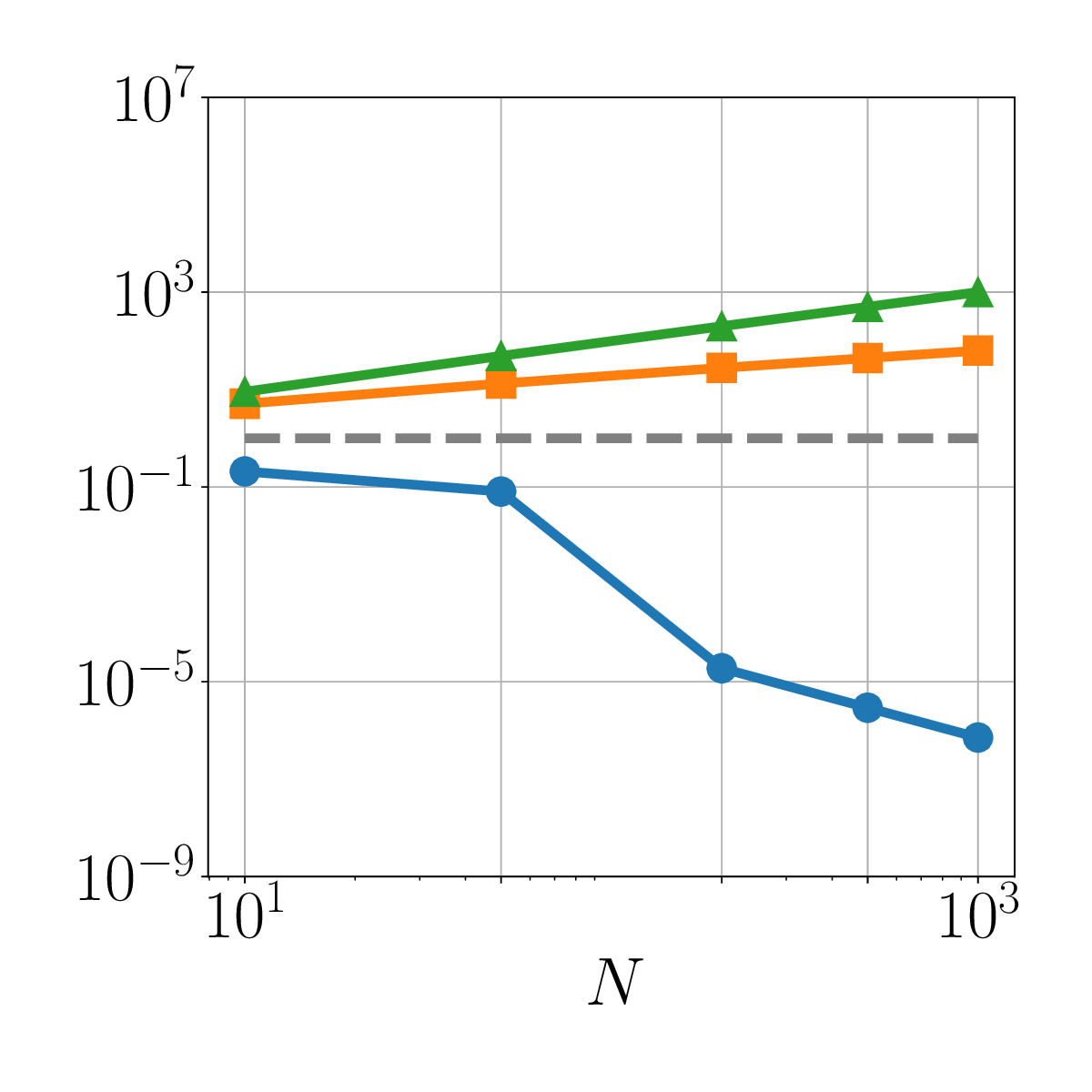}
        \caption{Complete network with random signs, $\epsilon(\numplayers) = \delta(\numplayers) = 1/\numplayers^3$}
    \end{subfigure}

    \begin{subfigure}[t]{0.25\textwidth}
        \centering
        \includegraphics[width=\linewidth]{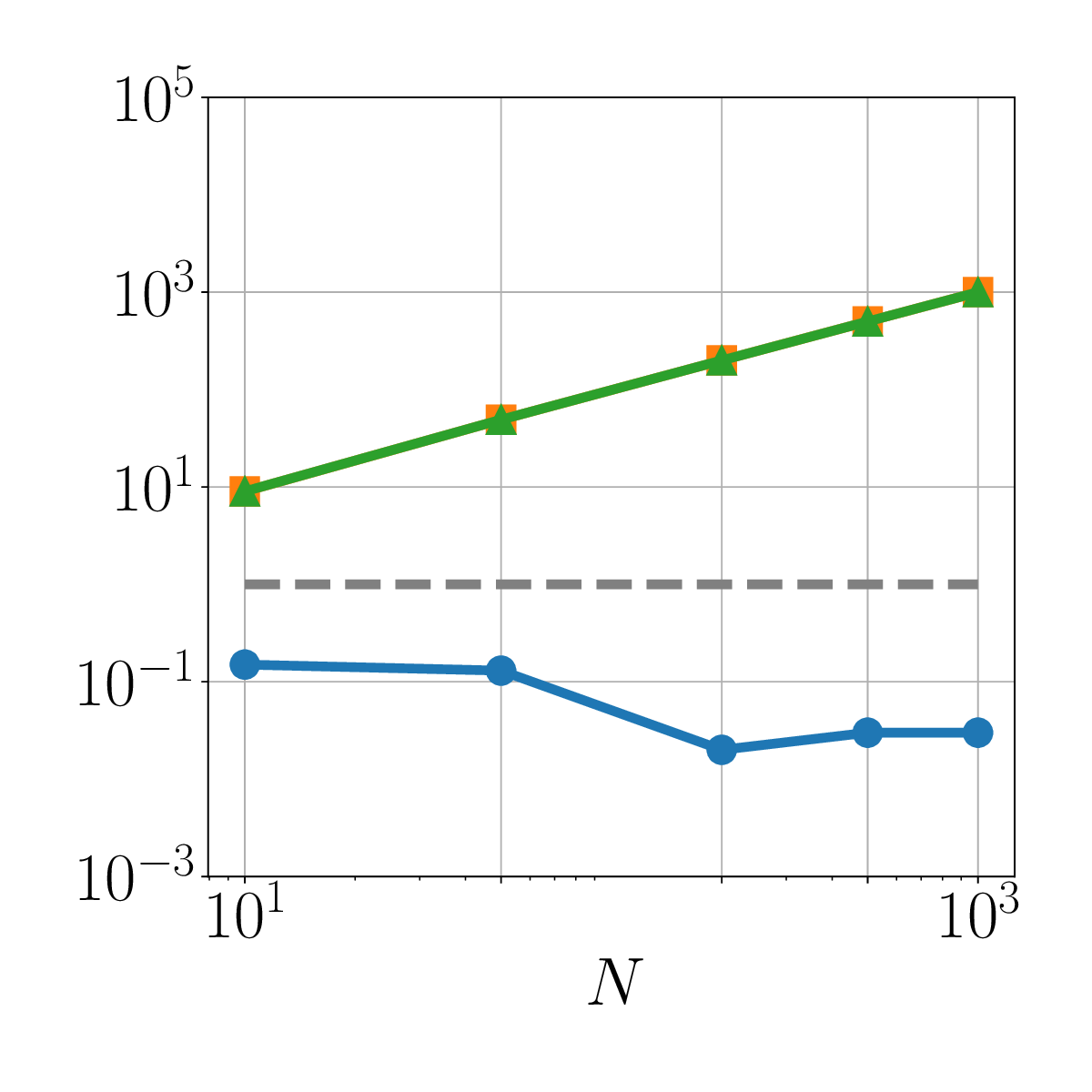}
        \caption{\erdosrenyi network, $p(\numplayers) = 1 - 1/\numplayers^{2.5}$}
    \end{subfigure}
    \hfill
    \begin{subfigure}[t]{0.25\textwidth}
        \centering
        \includegraphics[width=\linewidth]{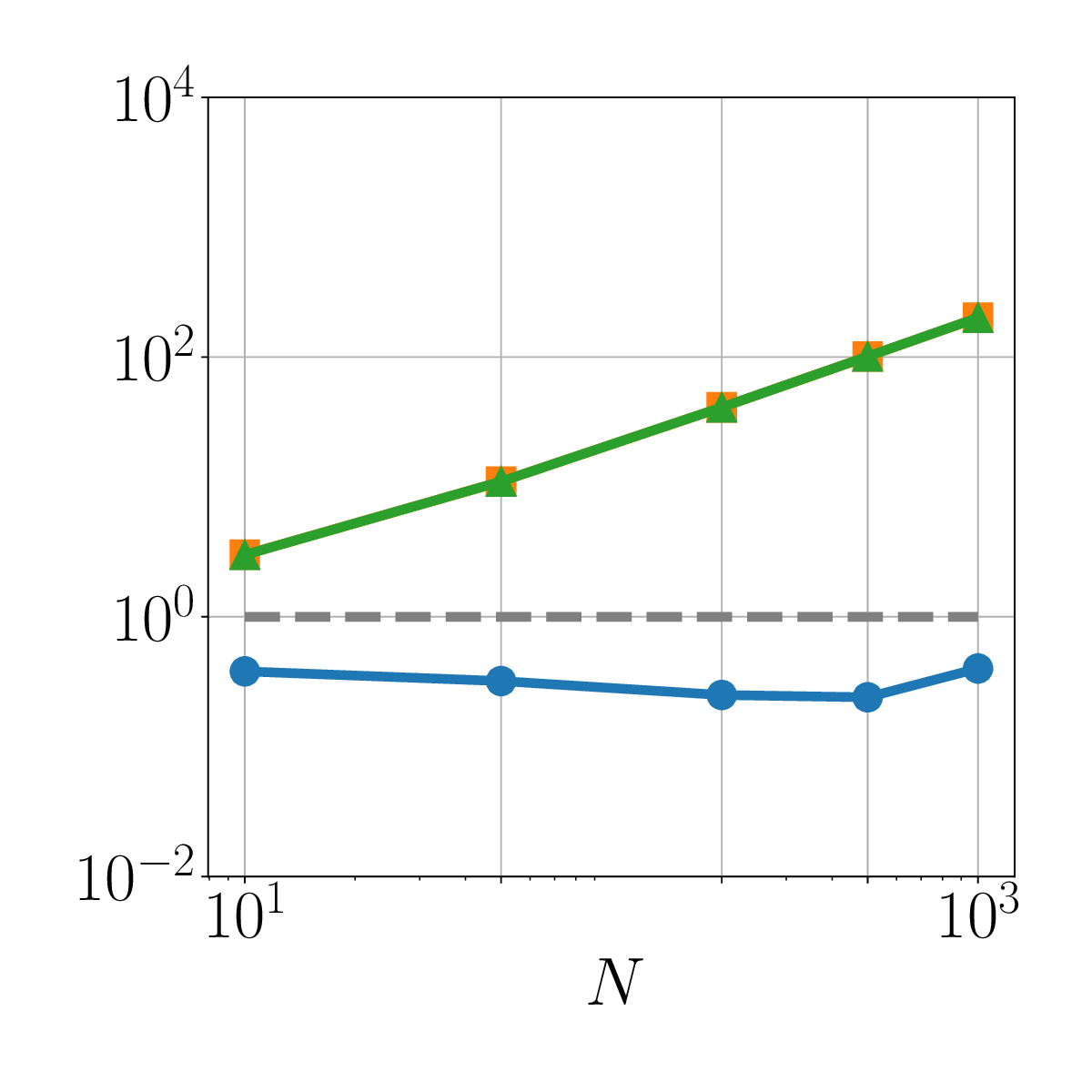}
        \caption{Small world network, $d(\numplayers) = \numplayers/10$, $p(\numplayers) = 1/\numplayers^2$}
    \end{subfigure}%
    \hfill
    \begin{subfigure}[t]{0.25\textwidth}
        \centering
        \includegraphics[width=\linewidth]{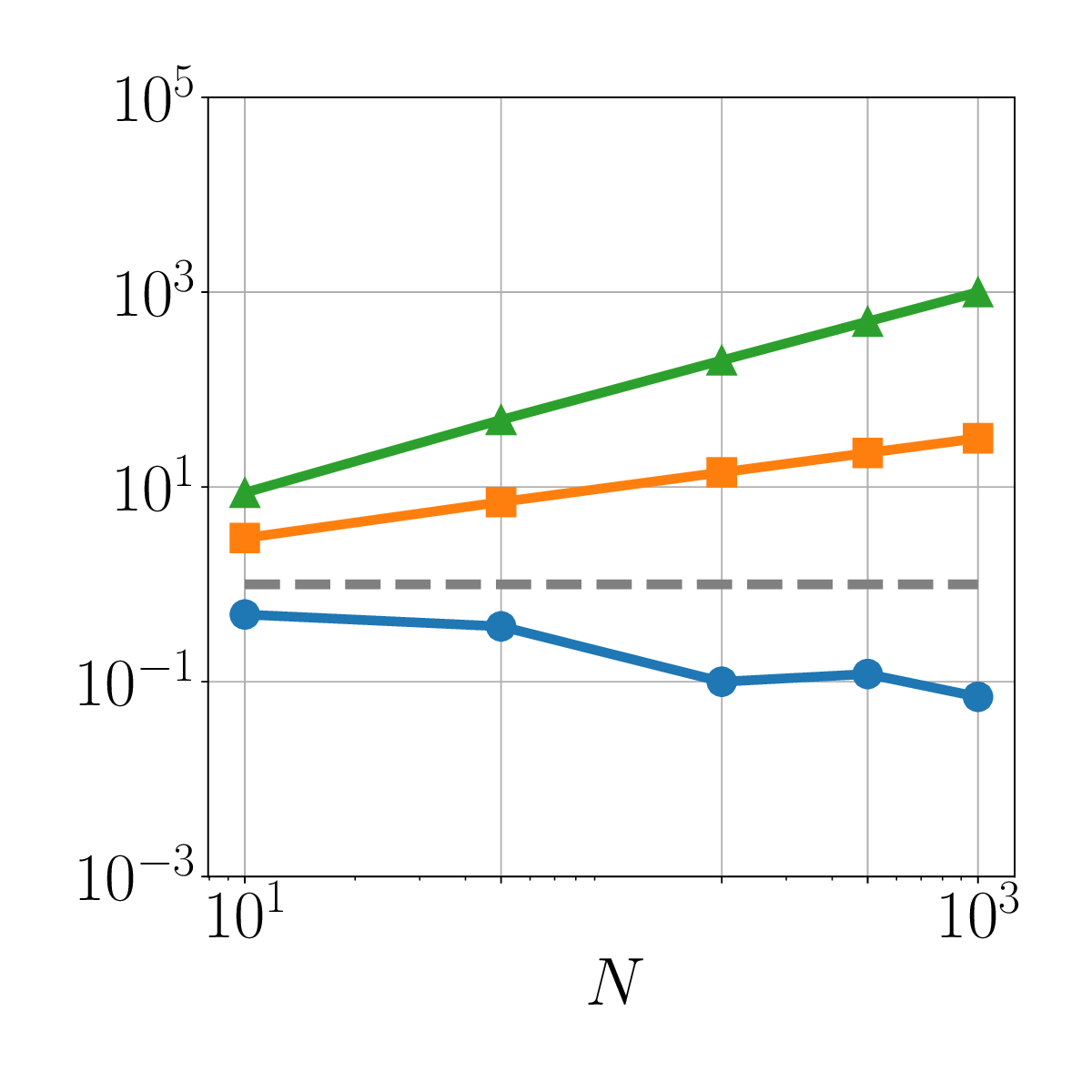}
        \caption{Star network with erased edges, $p(\numplayers) = 1/\numplayers^{1.5}$}
    \end{subfigure}%
    
    \begin{subfigure}{0.5\textwidth}
        \includegraphics[width=\linewidth]{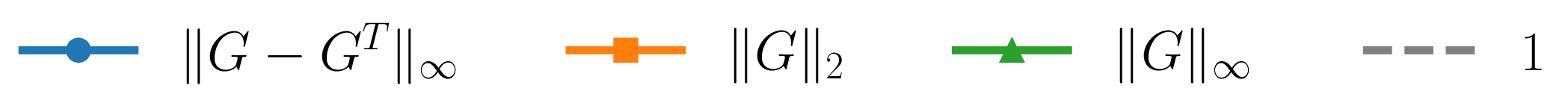}
    \end{subfigure}

    \caption{A plot of $\alpha = \|\adjacency - \adjacency^T\|_\infty$, $\|\adjacency\|_2$, and $\|\adjacency\|_\infty$ vs. the number of nodes $\numplayers \in \{10, 50, 200, 500, 1000\}$ for the network game and the network models described in Section \ref{sec: alpha for various networks}, averaged over $100$ instances of each network model. It is observed that $\alpha = \|\adjacency - \adjacency^T\|_\infty$ scales well with $\numplayers$, whereas $\|\adjacency\|_2$ and $\|\adjacency\|_\infty$ are greater than $1$. Thus, the framework of $\alpha$-potential games can be used to learn approximate Nash equilibria across a diverse set of networks even as the number of players grows, whereas the conditions $\|\adjacency\|_2 < 1$ and $\|\adjacency\|_\infty < 1$ required in \cite{francesca_VI_network_games} to ensure convergence of best-response dynamics are violated. 
    }
    \label{fig: alpha for various networks}
\end{figure}

\subsection{Welfare suboptimality}
\label{sec: price of stability}
 
In Section~\ref{sec: network games and alpha potential function}, we proved that Algorithms~\ref{alg: sequential best response} and \ref{alg: gradient play} ensure convergence of players' actions to the set of $2\alpha$-Nash equilibria. In this section, we focus on a particular $\alpha$-Nash equilibrium of the game: the maximizer of the $\alpha$-potential function given in \eqref{eq: LQ alpha potential function}.\footnote{An algorithm to learn the maximizer is discussed in Appendix \ref{sec: appendix alternative to gradient play algorithm}.} From the perspective of a central planner, our aim is to study the welfare properties of recommending such a solution to players in an LQ network game. 
Towards this, we adopt the total utility of all the players as the social welfare function. Specifically, define $\sw : \actionset \rightarrow \R$ such that
\begin{align}
\label{eq: social welfare}
    \sw(\actionvector) &:= \sum_{i = 1}^\numplayers \utility_\playeridx(\actionvector) = -\frac{1}{2} \sum_{i = 1}^\numplayers a_i^2 + \sum_{i = 1}^\numplayers \beta_i a_i + \gamma \sum_{i = 1}^\numplayers \left(\sum_{j = 1}^\numplayers G_{ij} a_j\right) a_i. 
\end{align}
Note that one can study a generalized social welfare function with convex combination of the utilities, however we consider a uniform combination to keep the exposition simple. 
We make the following assumption on the parameters of the LQ game to ensure that $\sw$ has a unique maximizer and to find its closed-form expression.

\begin{assumption}{(Contraction)}
\label{assn: contraction}
    Assume $|\gamma| < 1/2$, and for all $i, j \in [\numplayers], |G_{ij}| \leq 1/\numplayers$. Further, assume that $\exists \actiontilde \geq 0$ such that $[-\actiontilde, \actiontilde] \subseteq \actionset_i$ for all $i \in [\numplayers]$, and let $\betabar := \max_{i \in [\numplayers]} |\beta_i| \leq (1 - 2 |\gamma|) \actiontilde$. 
\end{assumption}

\begin{remark}
    Assumption \ref{assn: contraction} implies $\|\adjacency\|_2 \leq 1$, $\|\adjacency\|_\infty \leq 1$ and therefore is a sufficient condition for the game to be strongly monotone (see Lemma \ref{lem: lambda max G mat} and \cite{francesca_VI_network_games} for related results), and thus have a unique Nash equilibrium learnable via exact gradient dynamics \cite{facchinei2003finite}. 
\end{remark}

\begin{theorem}{(Welfare suboptimality)}
\label{thm: pos}
     Consider an LQ network game satisfying Assumption~\ref{assn: contraction}. Let $\actionvectoropt = \argmax{\actionvector \in \actionset} \sw(\actionvector)$ be the socially optimal action profile for $\sw$ in \eqref{eq: social welfare}, and $\actionvectorpotential = \argmax{\actionvector \in \actionset} \potentialfunction(\actionvector)$ be the maximizer of the $\alpha$-potential function in \eqref{eq: LQ alpha potential function}. Then, 
    \begin{align*}
        \frac{\sw(\actionvectoropt)}{\sw(\actionvectorpotential)} &\leq \underbrace{\frac{(1 - \lambdamin{\gamma(\adjacency + \adjacency^T)} \left(1 - \frac{\lambdamin{\gamma(\adjacency + \adjacency^T)}}{2}\right)^2}{\left(1 - \lambdamax{\gamma(\adjacency + \adjacency^T)}\right)^2}}_{\posboundeigvals} \\
        & \leq \underbrace{\frac{\left(1 + |\gamma| \max_{i \in [\numplayers]} \left(\sum_{j \neq i} |G_{ij}|+\sum_{j \neq i} |G_{ji}|\right)\right)^3}{\left(1 - |\gamma| \max_{i \in [\numplayers]} \left(\sum_{j \neq i} |G_{ij}| + \sum_{j \neq i} |G_{ji}|\right)\right)^2}}_{\posboundnetworkvals} \leq \underbrace{\frac{\left(1+2|\gamma|\right)^3}{\left(1-2|\gamma|\right)^2}}_{\posboundgamma}. 
     \end{align*}
\end{theorem}

The proof of Theorem~\ref{thm: pos} given in Appendix~\ref{sec: proof of theorem pos} proceeds by finding analytical expressions for $\actionvectoropt$ and $\actionvectorpotential$ under Assumption~\ref{assn: contraction}. Then, we substitute these expressions in the formula for $\sw$ in \eqref{eq: social welfare}, and use properties of the minimum and maximum eigenvalues of matrices. 
The result provides three types of bounds on the welfare suboptimality of $\actionvectorpotential$: the first bound $\posboundeigvals$ depends on $\gamma$ and the eigenvalues of the network adjacency matrix~$\adjacency$, and is the sharpest, the second bound $\posboundnetworkvals$ depends on $\gamma$ and the entries of $\adjacency$, while the third bound $\posboundgamma$ depends only on $\gamma$ and is the weakest. See Appendix \ref{sec: appendix numerical example for pos} for an illustration.

\begin{remark}
    For symmetric network games, it is known that $\potentialfunction$ in \eqref{eq: LQ alpha potential function} is an exact potential function of the LQ network game \cite{bramoulle_strategic_interactions} (note that $\alpha = 0$ in \eqref{eq: alpha bound LQ game} if $\adjacency = \adjacency^T$). Thus, under Assumption~\ref{assn: contraction}, $\actionvectorpotential$ is the unique Nash equilibrium of the network game (see also Lemma \ref{lem: unique alpha potential maximum} in Appendix~\ref{sec: appendix intermediate results}). Hence, the bounds derived in Theorem~\ref{thm: pos} are then bounds on the price of anarchy (PoA) of the network game.
\end{remark}

\section{Numerical results}
\label{sec: numerical results}

In this section, we consider a simple LQ network game over a complete network with random signs (as described in Example~\ref{ex: signed complete network}). We illustrate that the value of $\alpha$ for such a network is small. Further, we numerically verify the convergence of Algorithms~\ref{alg: sequential best response} and \ref{alg: gradient play}. In comparison, we demonstrate that the classical best-response and gradient play algorithms do not necessarily converge within 10,000 iterations. To verify the efficacy of the learned actions, we compare the social welfare of these algorithms.

\textbf{Game:} Consider an LQ network game with utility function as in \eqref{eq: LQ utility}. Let $\beta_i = 0$ for all $i \in [\numplayers]$ and $\gamma = 1$. Let the action sets of all players $i \in [\numplayers]$ be $\actionset_i = [-1, 1]$. This game can model scenarios where players' actions represent their propensity towards two alternative options represented by $\{-1, 1\}$ (such as two opposite activities or two extreme opinions).

\textbf{Network:} We consider a complete network with random signs of $\numplayers = 100$ players generated as described in Example \ref{ex: signed complete network} for $\epsilon(\numplayers) = \delta(\numplayers) = 1/\numplayers^3$. 
We generate $T = 100$ random instances (or trials) of the network. The values of the key metrics related to $\adjacency$ averaged over $T$ instances are given in Table \ref{tab: G norm values simulations}. 
\begin{table}[t]
    \centering
    \renewcommand{\arraystretch}{1.2} 
    \setlength{\tabcolsep}{11pt}
    \begin{tabular}{|c|c|c|}
    \hline
       $\alpha = \|\adjacency-\adjacency^\transpose\|_\infty$  & $\|\adjacency\|_2$ & $\|\adjacency\|_\infty$  \\ \hline
       $7.8 \times 10^{-5}$ &$19.31$ &$99.00$ \\ \hline
    \end{tabular}
    \caption{Average value of $\alpha = \|\adjacency-\adjacency^\transpose\|_\infty$ (see Proposition \ref{prop: alpha potential expression LQ}), $ \|\adjacency\|_2$ and $\|\adjacency\|_\infty$ for the complete network with random signs considered in Section \ref{sec: numerical results}, averaged over $T = 100$ trials. }
    \label{tab: G norm values simulations}
\end{table}
It is observed that the game does not usually have a potential function, since $\adjacency \neq \adjacency^T$ in general, nor does it usually satisfy $\|\adjacency\|_2 < 1$ or $\|\adjacency\|_\infty < 1$ as required in \cite{francesca_VI_network_games} to ensure convergence of best-response dynamics (see also Fig. \ref{fig: alpha for various networks}(c)). Moreover, the game has a mix of strategic complements ($G_{ij} > 0$) and substitutes ($G_{ij} < 0$). On the other hand, as shown in Example \ref{ex: signed complete network}, for any fixed $\delta > 0$ and $t > 0$, for $\numplayers$ large enough, with probability $1 - \delta$, $\alpha = \|\adjacency - \adjacency^T\|_\infty \leq 4/\numplayers^2 + t$.

\textbf{Algorithms and learned actions:} For each of the $T = 100$ instances of the network $\adjacency$ described above, we run Algorithm~\ref{alg: sequential best response} (with $\epsilon = \alpha$) and Algorithm~\ref{alg: gradient play} with the initial action of each player chosen uniformly at random from the player's action set. We also run the classical exact versions of these algorithms ($\alpha = 0$) to benchmark our results. 
In Table \ref{tab: algorithms comparison}, we observe that, in each of the $T$ trials, both Algorithms~\ref{alg: sequential best response} and \ref{alg: gradient play} reach $2\alpha$-Nash equilibria within a finite number of iterations, as guaranteed by Theorems~\ref{thm: sequential best response convergence} and \ref{thm: gradient play convergence}. 
On the other hand, the exact sequential best-response algorithm (Algorithm~\ref{alg: sequential best response} with $\alpha = \epsilon = 0$) and the exact simultaneous gradient play algorithm (Algorithm~\ref{alg: gradient play} with $\alpha = 0$) either take more iterations to converge, or do not converge within 10,000 iterations. Moreover, the social welfare of the actions learned by Algorithms~\ref{alg: sequential best response} and \ref{alg: gradient play} are comparable to those learned by the exact algorithms. 
\begin{table}[t]
    \centering
    \renewcommand{\arraystretch}{1.2} 
    \setlength{\tabcolsep}{2pt}
    \begin{tabular}{|c|c|c|c|}
        \hline
        \textbf{Algorithm} & \textbf{\% Convergence} & \textbf{\# Iterations} & \textbf{Social Welfare}  \\ \hline
        Best-response with $\alpha > 0$ (Algorithm~\ref{alg: sequential best response})  & 100\% & 851.25 & 1324.76 \\ \hline
        Gradient play with $\alpha > 0$ (Algorithm~\ref{alg: gradient play})  & 100\% & 100.10 & 1268.89 \\ \hline
        Exact best-response (Algorithm \ref{alg: sequential best response}, $\alpha = \epsilon = 0$) & 89\% & 1827.57 & 1274.40 \\ \hline
        Exact gradient play (Algorithm \ref{alg: gradient play}, $\alpha = 0$)  & 95\% & 109.57 & 1270.08 \\ \hline
    \end{tabular}
    \caption{Performance of various algorithms for the LQ network game described in Section \ref{sec: numerical results}. The numbers represent the mean over $T = 100$ instances of the network game. The first column is the \% of trials in which the algorithms converged within 10,000 iterations (we say that an algorithm converges at iteration $k$ if for all $i \in [\numplayers]$, $|a^{k+1}_i - a^k_i| \leq 10^{-8} + 10^{-5} \times a^k_i$).  The second column is the average number of iterations to convergence, and the third column is the average social welfare at termination, averaged over the trials where all algorithms converged within 10,000 iterations. }
    \vspace{-10mm}
    \label{tab: algorithms comparison}
\end{table}
Additional numerical results illustrating the quality of the learned strategies with respect to the utilities of the players, and the monotonic increase of the $\alpha$-potential function are provided in Appendix \ref{sec: appendix additional numerical results}.

\section{Conclusion}
\label{sec: conclusion}

We provided an explicit characterization of an $\alpha$-potential function for static games under mild regularity assumptions and showed that the $\alpha$-potential function can be used to guide learners to the set of approximate Nash equilibria. In particular, we proposed variants of the sequential best-response algorithm and the gradient play algorithm dependent on the value of $\alpha$ that converge to the set of 2$\alpha$-Nash equilibria in finite time. Our convergence results hold for network games without restrictive assumptions such as symmetry of the network or bounds on the norms of the network adjacency matrix. For linear-quadratic network games, we showed that $\alpha$ is proportional to the degree of asymmetry in the network quantified by the infinity norm of the antisymmetric component of the network adjacency matrix. We showed that in several network topologies, $\alpha$ scales well with the size of the network. Finally, we derived bounds on the social welfare of the particular $\alpha$-Nash equilibrium given by the maximizer of the $\alpha$-potential function. Through numerical experiments, we demonstrated that the proposed algorithms consistently reach approximate equilibria while standard algorithms may fail to converge. Thus, the framework of $\alpha$-potential games broadens the family of network games for which simple payoff-driven learning can be used to guarantee convergence to approximate Nash equilibria. Future extensions can look at time-varying games, stochastic algorithms, and network games with a community structure. Computing the optimal $\alpha$-potential function for a broad class of games is also an interesting problem.

\begin{credits}

\subsubsection{\ackname} This material is based upon work supported by the National Science Foundation under Award No. ECCS-2340289. Vikram Krishnamurthy and Adit Jain were supported by National Science Foundation grants CCF-2312198 and CCF-2112457. Eva Tardos was supported in part by AFOSR grants FA9550-23-1-0410 and FA9550-231-0068, and ONR MURI grant N000142412742. 

\end{credits}

\bibliographystyle{splncs04}
\bibliography{refs}

\begin{thebibliography}{10}
\providecommand{\url}[1]{\texttt{#1}}
\providecommand{\urlprefix}{URL }
\providecommand{\doi}[1]{https://doi.org/#1}

\bibitem{Ballester_2006}
Ballester, C., Calvó-Armengol, A., Zenou, Y.: Who's who in networks. {W}anted:
  The key player. Econometrica  \textbf{74}(5),  1403--1417 (2006)

\bibitem{BINDEL2015Opinion}
Bindel, D., Kleinberg, J., Oren, S.: How bad is forming your own opinion? Games
  and Economic Behavior  \textbf{92},  248--265 (2015)

\bibitem{bramoulle_strategic_interactions}
Bramoullé, Y., Kranton, R., D'Amours, M.: Strategic interaction and networks.
  American Economic Review  \textbf{104}(3),  898--930 (2014)

\bibitem{bubeck2015convexopt}
Bubeck, S.: Convex optimization: Algorithms and complexity. Found. Trends Mach.
  Learn.  \textbf{8}(3–4),  231–357 (Nov 2015)

\bibitem{bullo2018lectures}
Bullo, F.: Lectures on Network Systems. Kindle Direct Publishing, {1.7} edn.
  (2024)

\bibitem{candogan_optimal_pricing}
Candogan, O., Bimpikis, K., Ozdaglar, A.: Optimal pricing in networks with
  externalities. Operations Research  \textbf{60}(4),  883--905 (2012)

\bibitem{Candogan_CDC}
Candogan, O., Ozdaglar, A., Parrilo, P.A.: A projection framework for
  near-potential games. In: 49th IEEE Conference on Decision and Control (CDC).
  pp. 244--249 (2010)

\bibitem{CANDOGAN_GEB}
Candogan, O., Ozdaglar, A., Parrilo, P.A.: Dynamics in near-potential games.
  Games and Economic Behavior  \textbf{82},  66--90 (2013)

\bibitem{Candogan_TEAC}
Candogan, O., Ozdaglar, A., Parrilo, P.A.: Near-potential games: Geometry and
  dynamics. ACM Trans. Econ. Comput.  \textbf{1}(2) (2013)

\bibitem{facchinei2003finite}
Facchinei, F., Pang, J.S.: Finite-dimensional variational inequalities and
  complementarity problems. Springer (2003)

\bibitem{Pavel_passivity_networks_TAC}
Gadjov, D., Pavel, L.: A passivity-based approach to {N}ash equilibrium seeking
  over networks. IEEE Transactions on Automatic Control  \textbf{64}(3),
  1077--1092 (2019)

\bibitem{Gadjov2021Aggregative}
Gadjov, D., Pavel, L.: Single-timescale distributed {GNE} seeking for
  aggregative games over networks via forward–backward operator splitting.
  IEEE Transactions on Automatic Control  \textbf{66}(7),  3259--3266 (2021)

\bibitem{galeotti_interventions}
Galeotti, A., Golub, B., Goyal, S.: Targeting interventions in networks.
  Econometrica  \textbf{88}(6),  2445--2471 (2020)

\bibitem{network_games_galeotti}
Galeotti, A., Goyal, S., Jackson, M., Vega-Redondo, F., Yariv, L.: Network
  games. The Review of Economic Studies  \textbf{77}(1),  218--244 (2010)

\bibitem{Grammatico2018Proximal}
Grammatico, S.: Proximal dynamics in multiagent network games. IEEE
  Transactions on Control of Network Systems  \textbf{5}(4),  1707--1716 (2018)

\bibitem{guo2025markovalphapotentialgames}
Guo, X., Li, X., Maheshwari, C., Sastry, S., Wu, M.: Markov $\alpha$-potential
  games. arXiv preprint, arXiv:2305.12553  (2025)

\bibitem{guo2025alphapotentialgameframeworknplayer}
Guo, X., Li, X., Zhang, Y.: An $\alpha$-potential game framework for $n$-player
  dynamic games. arXiv preprint, arXiv:2403.16962  (2025)

\bibitem{jackson_behavioral_communities}
Jackson, M.O., Storms, E.C.: Behavioral communities and the atomic structure of
  networks. arXiv preprint, arXiv:1710.04656  (2023)

\bibitem{jackson_chapter}
Jackson, M.O., Zenou, Y.: Chapter 3 - {G}ames on networks. In: Connections: An
  Introduction to the Economics of Networks, Handbook of Game Theory with
  Economic Applications, vol.~4, pp. 95--163. Elsevier (2015)

\bibitem{kearns_2001_graphical_games}
Kearns, M., Littman, M.L., Singh, S.: Graphical models for game theory. In:
  Proceedings of the Seventeenth Conference on Uncertainty in Artificial
  Intelligence. p. 253–260. UAI'01, Morgan Kaufmann Publishers Inc., San
  Francisco, CA, USA (2001)

\bibitem{kempe_diffusion_model}
Kempe, D., Kleinberg, J., Tardos, {\'E}.: Influential nodes in a diffusion
  model for social networks. In: Automata, Languages and Programming. pp.
  1127--1138. Springer Berlin Heidelberg (2005)

\bibitem{Liu_2003_sensor_networks_asymmetric}
Liu, J., Li, B.: Distributed topology control in wireless sensor networks with
  asymmetric links. In: IEEE Global Telecommunications Conference (GLOBECOM).
  vol.~3, pp. 1257--1262 (2003)

\bibitem{Matni_CDC_projection_framework}
Matni, N.: A projection framework for near-potential polynomial games. In: 51st
  IEEE Conference on Decision and Control (CDC). pp. 6507--6512 (2012)

\bibitem{Mazumdar_2020}
Mazumdar, E., Ratliff, L.J., Sastry, S.S.: On gradient-based learning in
  continuous games. SIAM Journal on Mathematics of Data Science  \textbf{2}(1),
   103–131 (2020)

\bibitem{shapley_potential}
Monderer, D., Shapley, L.S.: Potential games. Games and Economic Behavior
  \textbf{14}(1),  124--143 (1996)

\bibitem{nocedal1999numerical}
Nocedal, J., Wright, S.J.: Numerical optimization. Springer (1999)

\bibitem{paccagnan_aggregative_games}
Paccagnan, D., Gentile, B., Parise, F., Kamgarpour, M., Lygeros, J.: {N}ash and
  {W}ardrop equilibria in aggregative games with coupling constraints. IEEE
  Transactions on Automatic Control  \textbf{64}(4),  1373--1388 (2019)

\bibitem{PAPADIMITRIOU_public_good_directed_network_games}
Papadimitriou, C., Peng, B.: Public goods games in directed networks. Games and
  Economic Behavior  \textbf{139},  161--179 (2023)

\bibitem{PARISE2020Distributed}
Parise, F., Grammatico, S., Gentile, B., Lygeros, J.: Distributed convergence
  to {N}ash equilibria in network and average aggregative games. Automatica
  \textbf{117},  108959 (2020)

\bibitem{francesca_VI_network_games}
Parise, F., Ozdaglar, A.: A variational inequality framework for network games:
  Existence, uniqueness, convergence and sensitivity analysis. Games and
  Economic Behavior  \textbf{114},  47--82 (2019)

\bibitem{francesca_review}
Parise, F., Ozdaglar, A.: Analysis and interventions in large network games.
  Annual Review of Control, Robotics, and Autonomous Systems  \textbf{4}(1),
  455--486 (2021)

\bibitem{asymmetricallycommittedrel}
Stanley, S.M., Rhoades, G.K., Scott, S.B., Kelmer, G., Markman, H.J., Fincham,
  F.D.: Asymmetrically committed relationships. Journal of Social and Personal
  Relationships  \textbf{34}(8),  1241--1259 (2017)

\bibitem{feras_learning_LQ_games}
Taha, F.A., Rokade, K., Parise, F.: Gradient dynamics in linear quadratic
  network games with time-varying connectivity and population fluctuation. In:
  62nd IEEE Conference on Decision and Control (CDC). pp. 1991--1996 (2023)

\bibitem{VANELLI20201Coordinating}
Vanelli, M., Arditti, L., Como, G., Fagnani, F.: On games with coordinating and
  anti-coordinating agents. IFAC-PapersOnLine  \textbf{53}(2),  10975--10980
  (2020)

\bibitem{watts1998small-world}
Watts, D.J., Strogatz, S.H.: Collective dynamics of ‘small-world’ networks.
  Nature  \textbf{393}(6684),  440--442 (1998)

\bibitem{ZHU2022Aggregative}
Zhu, R., Zhang, J., You, K., Başar, T.: Asynchronous networked aggregative
  games. Automatica  \textbf{136},  110054 (2022)

\end{thebibliography}

\newpage

\begin{center}
    {\large \textbf{APPENDIX}}
\end{center}

\appendix
\vspace{1ex}

\section{Proof of Corollary \ref{cor: alpha potential existence}}
\label{sec: proof of theorem alpha potential existence}

\textbf{Key insight:} The goal of Corollary~\ref{cor: alpha potential existence} is to establish that, under mild smoothness and compactness assumptions on the game, the proposed integral expression (given in equation~\eqref{eq: alpha potential expression}) is an $\alpha$-potential function of the game. The key idea of the proof is to quantify the first-order difference between the utility function and the proposed $\alpha$-potential function (see the left-hand-side of equation~\eqref{eq: alpha potential condition} in Definition~\ref{def: alpha potential function}) in terms of the second-order partial derivatives of the players' utilities, which in-turn gives the bound given in equation~\eqref{eq: alpha bound}. The proof specializes a known result for Markov games \cite[Theorem~2.5]{guo2025alphapotentialgameframeworknplayer} to the static setting, and verifies the $\alpha$-potential property. 

\vspace{1ex}
\noindent\textbf{Step 1: Derivative of $\potentialfunction$: }
Since $\utility_\playeridx \in \mathcal{C}^2(\R^\numplayers)$, the function $\potentialfunction$ is well defined and differentiable. Using Leibniz rule, we obtain that
\begin{align}
\label{eq: partial derivative potential function}
    \frac{\partial \potentialfunction}{\partial \action_\playeridx}(\actionvector) = &\int_0^1 \frac{\partial \utility_\playeridx}{\partial \action_\playeridx}(\actionvectorfixed + r(\actionvector - \actionvectorfixed)) \, dr + \int_0^1 \sum_{\playeridxalt=1}^\numplayers \frac{\partial^2 \utility_\playeridxalt}{\partial \action_\playeridxalt \partial \action_\playeridx}(\actionvectorfixed + r(\actionvector - \actionvectorfixed)) \cdot r \cdot (\action_\playeridxalt - \actionfixed_\playeridxalt) \, dr. 
\end{align}

\vspace{1ex}
\noindent\textbf{Step 2: Relation between derivatives of $\potentialfunction$ and $\utility_\playeridx$: }
Let us define the error  
\begin{align*}
    \mathcal{E}_{r}^i (\actionvector, \actionvectorfixed) = \sum_{\playeridxalt=1}^\numplayers \left( \frac{\partial^2 \utility_\playeridxalt}{\partial \action_\playeridxalt \partial \action_\playeridx}(\actionvectorfixed + r(\actionvector - \actionvectorfixed)) - \frac{\partial^2 \utility_\playeridx}{\partial \action_\playeridx \partial \action_\playeridxalt}(\actionvectorfixed + r(\actionvector - \actionvectorfixed)) \right) \cdot (\action_\playeridxalt - \actionfixed_\playeridxalt) 
\end{align*}
so that equation \eqref{eq: partial derivative potential function} can be written as 
\begin{align*}
    \frac{\partial \potentialfunction}{\partial \action_\playeridx}(\actionvector) =\int_0^1 \frac{\partial \utility_\playeridx}{\partial \action_\playeridx}(\actionvectorfixed + r(\actionvector - \actionvectorfixed)) \, dr + \int_0^1 \mathcal{E}_{r}^i (\actionvector, \actionvectorfixed) rdr  + \int_0^1 \sum_{\playeridxalt=1}^\numplayers \frac{\partial^2 \utility_\playeridx}{\partial \action_\playeridx \partial \action_\playeridxalt}(\actionvectorfixed + r(\actionvector - \actionvectorfixed)) \cdot r \cdot (\action_\playeridxalt - \actionfixed_\playeridxalt) \, dr. 
\end{align*}
Using the chain rule on the partial derivatives, we obtain that 
\begin{align*}
    \frac{d}{dr}\frac{\partial \utility_\playeridx}{\partial \action_\playeridx}(\actionvectorfixed + r(\actionvector - \actionvectorfixed)) = \sum_{\playeridxalt=1}^\numplayers \frac{\partial^2 \utility_\playeridx}{\partial \action_\playeridx \partial \action_\playeridxalt}(\actionvectorfixed + r(\actionvector - \actionvectorfixed)) \cdot (\action_\playeridxalt - \actionfixed_\playeridxalt).
\end{align*}
Moreover, it follows from integration by parts that
\begin{align*}
    \int_0^1 \frac{\partial \utility_\playeridx}{\partial \action_\playeridx}(\actionvectorfixed + r(\actionvector - \actionvectorfixed)) \, dr +  \int_0^1 \frac{d}{dr}\frac{\partial \utility_\playeridx}{\partial \action_\playeridx}(\actionvectorfixed + r(\actionvector - \actionvectorfixed)) \cdot r dr  = \left[r\frac{\partial \utility_\playeridx}{\partial \action_\playeridx}(\actionvectorfixed + r(\actionvector - \actionvectorfixed) )\right]_0^1 = \frac{\partial \utility_\playeridx}{\partial \action_\playeridx}(\actionvector). 
\end{align*}
Hence, 
\begin{align}
\label{eq: phi and u_i derivative first relation}
    \frac{\partial \potentialfunction}{\partial \action_\playeridx}(\actionvector)  = \frac{\partial \utility_\playeridx}{\partial \action_\playeridx}(\actionvector) + \int_0^1 \mathcal{E}_{r}^i (\actionvector, \actionvectorfixed) rdr. 
\end{align}

\vspace{1ex}
\noindent\textbf{Step 3: $\alpha$-potential property: }
For any $\playeridx \in \numplayers$, $\actionvector \in \actionset$, and $\action_\playeridx^\prime \in \actionset_\playeridx$, using the fundamental theorem of calculus, we obtain that
\begin{align*}
    \utility_\playeridx((\action_\playeridx^\prime, \action_{-\playeridx})) - \utility_\playeridx((\action_\playeridx, \action_{-\playeridx})) = \int_0^1 \frac{\partial \utility_\playeridx}{\partial \action_\playeridx}(\action_\playeridx + \varepsilon(\action_\playeridx^\prime - \action_\playeridx), \action_{-\playeridx}) \cdot (\action_\playeridx^\prime - \action_\playeridx) \, d\varepsilon. 
\end{align*}
Similarly, for the potential function, 
\begin{align*}
    \potentialfunction((\action_\playeridx^\prime, \action_{-\playeridx})) - \potentialfunction((\action_\playeridx, \action_{-\playeridx})) = \int_0^1 \frac{\partial \potentialfunction}{\partial \action_\playeridx}(\action_\playeridx + \varepsilon(\action_\playeridx^\prime - \action_\playeridx), \action_{-\playeridx}) \cdot (\action_\playeridx^\prime - \action_\playeridx) \, d\varepsilon. 
\end{align*}
Substituting the expression from \eqref{eq: phi and u_i derivative first relation} and defining $\actionvector^\varepsilon = (\action_\playeridx + \varepsilon(\action_\playeridx^\prime - \action_\playeridx), \action_{-\playeridx})$
we have
\begin{align*}
    \potentialfunction((\action_\playeridx^\prime, \action_{-\playeridx})) - \potentialfunction((\action_\playeridx, \action_{-\playeridx})) = \int_0^1 \left[    \frac{\partial \utility_\playeridx}{\partial \action_\playeridx}(\actionvector^\varepsilon)  + \int_0^1 \mathcal{E}_{r}^i (\actionvector^\varepsilon, \actionvectorfixed) r dr \right] \cdot (\action_\playeridx^\prime - \action_\playeridx) \, d\varepsilon. 
\end{align*}
Thus, the difference is given by  
\begin{align}
\label{eq: utility and potential difference difference}
    \left\vert \utility_\playeridx(\action_\playeridx^\prime, \action_{-\playeridx}) - \utility_\playeridx(\action_\playeridx, \action_{-\playeridx})  - (\potentialfunction(\action_\playeridx^\prime, \action_{-\playeridx}) - \potentialfunction(\action_\playeridx, \action_{-\playeridx})) \right\vert =  \left \vert \int_0^1 \int_0^1  \mathcal{E}_{r}^i (\actionvector^\varepsilon, \actionvectorfixed) (\action_\playeridx^\prime - \action_\playeridx) rd\varepsilon dr \right\vert.
\end{align}

\vspace{1ex}
\noindent\textbf{Step 4: Bounding the error term: }
We first bound the difference above as
\begin{align*}
   \left \vert \int_0^1 \int_0^1  \mathcal{E}_{r}^i (\actionvector^\varepsilon, \actionvectorfixed) (\action_\playeridx^\prime - \action_\playeridx) rd\varepsilon dr \right \vert \leq \adelta \int_0^1 \int_0^1  \left \vert\mathcal{E}_{r}^i (\actionvector^\varepsilon, \actionvectorfixed) \right\vert rd\varepsilon dr
    &\leq \adelta\sup_{r,\varepsilon}\left\vert\mathcal{E}_{r}^i (\actionvector^\varepsilon, \actionvectorfixed) \right\vert  \int_0^1 \int_0^1   rd\varepsilon dr \\
    &= \frac{1}{2} \adelta\sup_{r,\varepsilon}\left\vert\mathcal{E}_{r}^i (\actionvector^\varepsilon, \actionvectorfixed) \right\vert. 
\end{align*}
Further, we bound the error term as 
\begin{align*}
    \vert\mathcal{E}_{r}^i (\actionvector^\varepsilon, \actionvectorfixed) \vert &\leq \sum_{\playeridxalt=1}^\numplayers \left\vert\left( \frac{\partial^2 \utility_\playeridxalt}{\partial \action_\playeridxalt \partial \action_\playeridx}(\actionvectorfixed + r(\actionvector^\varepsilon - \actionvectorfixed)) - \frac{\partial^2 \utility_\playeridx}{\partial \action_\playeridx \partial \action_\playeridxalt}(\actionvectorfixed + r(\actionvector^\varepsilon - \actionvectorfixed)) \right) \right\vert \cdot \left\vert  (\action_\playeridxalt - \actionfixed_\playeridxalt)\right\vert \\
    &\leq \adelta\max_{ \actionvector\in\actionset} \sum_{\playeridxalt=1}^\numplayers \left|\frac{\partial^2 \utility_\playeridx}{\partial \action_\playeridx \partial \action_\playeridxalt}(\actionvector) - \frac{\partial^2 \utility_\playeridxalt}{\partial \action_\playeridxalt \partial \action_\playeridx}(\actionvector)\right|. 
\end{align*}
Substituting these bounds into \eqref{eq: utility and potential difference difference}, we obtain that $\potentialfunction$ is an $\alpha$-potential function with
\begin{align*}
    \alpha = \frac{\adelta^2}{2} \max_{\playeridx \in [\numplayers], \actionvector\in\actionset} \sum_{\playeridxalt=1}^\numplayers \left|\frac{\partial^2 \utility_\playeridx}{\partial \action_\playeridx \partial \action_\playeridxalt}(\actionvector) - \frac{\partial^2 \utility_\playeridxalt}{\partial \action_\playeridxalt \partial \action_\playeridx}(\actionvector)\right|. 
\end{align*}
Further, from \eqref{eq: phi and u_i derivative first relation}, we obtain that, for all $a \in \actionset$, 
\begin{align*}
    \left|\frac{\partial \potentialfunction}{\partial \action_\playeridx}(\actionvector) - \frac{\partial \utility_\playeridx}{\partial \action_\playeridx}(\actionvector)\right| &\leq \left|\int_0^1 \mathcal{E}_{r}^i (\actionvector, \actionvectorfixed) rdr\right| \\
    &\leq \sup_{r \in [0,1], \actionvector \in \actionset} |\mathcal{E}_r^i (\actionvector, \actionvectorfixed)| \cdot \left|\int_0^1 rdr\right| \\
    &\leq \frac{\adelta}{2}\max_{ \actionvector\in\actionset} \sum_{\playeridxalt=1}^\numplayers \left|\frac{\partial^2 \utility_\playeridx}{\partial \action_\playeridx \partial \action_\playeridxalt}(\actionvector) - \frac{\partial^2 \utility_\playeridxalt}{\partial \action_\playeridxalt \partial \action_\playeridx}(\actionvector)\right| \\
    &\leq \frac{\alpha}{\adelta}. 
\end{align*}

\section{Proof of Theorem~\ref{thm: sequential best response convergence}}
\label{sec: proof of theorem sequential best-response convergence}

Consider any $k \in \{0 ,1, 2, \dots \}$. Suppose that there exists a player $i \in [\numplayers]$ which satisfies $\max_{a_i \in \actionset_i} \utility_\playeridx(a_i, \actionvector_{-i}^k) > \utility_\playeridx(\actionvector^k) + \alpha + \epsilon$. Then, 
\begin{align*}
    \potentialfunction(\actionvector^{k+1}) &= \potentialfunction(a_i^{k+1}, \actionvector^{k}_{-i}) \\
    &= \potentialfunction(a_i^{k+1}, \actionvector^{k}_{-i}) - \potentialfunction(\actionvector^k) + \potentialfunction(\actionvector^k) \\
    &\stackrel{(a)}{\geq} \utility_\playeridx(a_i^{k+1}, \actionvector^{k}_{-i}) - \utility_\playeridx(\actionvector^k) -  \alpha + \potentialfunction(\actionvector^k) \\
    &\stackrel{(b)}{>} \potentialfunction(\actionvector^k) + \epsilon
\end{align*}
where $(a)$ follows from the fact that $\potentialfunction$ is an $\alpha$-potential function of the game (see Definition~\ref{def: alpha potential function}), and $(b)$ follows from the update rule of player~$i$ in the algorithm which implies $\utility_\playeridx(a_i^{k+1}, \actionvector_{-i}^k) > \utility_\playeridx(\actionvector^k) + \alpha + \epsilon$. On the other hand, if there does not exist a player~$i$ which satisfies the condition $\max_{a_i \in \actionset_i} \utility_\playeridx(a_i, \actionvector_{-i}^k) > \utility_\playeridx(\actionvector^k) + \alpha + \epsilon$, then $\actionvector^{k + 1} = \actionvector^{k}$ and hence $\potentialfunction(\actionvector^{k + 1}) = \potentialfunction(\actionvector^{k})$. 
Since $\max_{\actionvector \in \actionset} |\potentialfunction(\actionvector)| = \potentialfunction^{\max} < \infty$, after at most $\bar{k} := \ceil{\frac{2 \potentialfunction^{\max}}{\epsilon}}$ number of iterations, there cannot exist a player~$i$ which satisfies the condition $\max_{a_i \in \actionset_i} \utility_\playeridx(a_i, \actionvector_{-i}^k) > \utility_\playeridx(\actionvector^k) + \alpha + \epsilon$. Hence, we must have $\actionvector^{k} = \actionvector^{\bar{k}}$ for all $k \geq \bar{k}$, and $\max_{a_i \in \actionset_i} \utility_\playeridx(a_i, \actionvector_{-i}^k) \leq \utility_\playeridx(\actionvector^k) + \alpha + \epsilon$ for all $i \in [\numplayers]$.

\section{Proof of Theorem~\ref{thm: gradient play convergence}}
\label{sec: proof of theorem simultenous gradient play}

\textbf{Key insight:} Theorem~\ref{thm: gradient play convergence} shows that in any game with a smooth $\alpha$-potential function, a suitably modified simultaneous gradient play algorithm will reach a $2\alpha$-Nash equilibrium in finitely many steps. Our result uses ideas from the convergence results for projected gradient ascent on smooth functions using inexact gradients. The key idea of the proof is to show that, since the players' incentives are approximately aligned with the $\alpha$-potential function (see Definition~\ref{def: alpha potential function}), the players updating their actions using local gradient information leads to a monotonic increase in the value of the $\alpha$-potential function until the action profile reaches a set of approximate Nash equilibria. Towards this, we show that, at each iteration as the players update their actions, either the $\alpha$-potential increases by a fixed positive amount, or the current action profile belongs to the set of all $2\alpha$-Nash equilibria of the game. The step-size constraint intuitively guarantees that the sequence of action profiles does not cycle indefinitely. 

\vspace{1ex}
\noindent\textbf{Proof outline:} First, note that since $\potentialfunction$ is $\lipschitz$-smooth, for all $k \geq 0$, 
\begin{align}
    \potentialfunction(\actionvector^{\timeindex+1}) &\stackrel{(a)}{\geq} \potentialfunction(\actionvector^\timeindex) + \nabla \potentialfunction(\actionvector^\timeindex)^T (\actionvector^{\timeindex+1} - \actionvector^\timeindex) - \frac{\lipschitz}{2} \|\actionvector^{\timeindex+1} - \actionvector^\timeindex\|_2^2 \nonumber \\
    &= \potentialfunction(\actionvector^\timeindex) + \sum_{i = 1}^\numplayers \underbrace{\left[\frac{\partial}{\partial a_i} \potentialfunction(\actionvector^k) (a^{\timeindex+1}_i - a^\timeindex_i) - \frac{\lipschitz}{2} (a^{\timeindex+1}_i - a^\timeindex_i)^2\right]}_{=: \potentialincrement^k_i} \label{eq: potential function increment L smooth property}
\end{align}
where $(a)$ follows from Lemma \ref{lem: smooth function property} in Appendix \ref{sec: appendix intermediate results}. Note that $\potentialincrement^k_i$ represents the change in the $\alpha$-potential function due to player~$i$'s updated action at time $k$. If player~$i$ does not update his action at time $k$, then $\potentialincrement^k_i = 0$. To lowerbound $\potentialincrement^k_i$ when player~$i$ updates his action at time $k$, we divide $\actionset_i$ into five $k$-dependent regions. 
Specifically, let $a_i^{+k} := a_i + \stepsize_i \partial/(\partial a_i)\utility_\playeridx(a_i, \actionvector^{\timeindex}_{-i})$, and
\begin{align*}
    \regiononea &:= \{a_i \in \actionset_i : a_i^{+k} \geq \actionmax_i, \ a_i \leq \actionmax_i - \constant_2 \alpha\}, \\
    \regiononeb &:= \{a_i \in \actionset_i : a_i^{+k} \leq \actionmin_i, \ a_i \geq \actionmin_i + \constant_2 \alpha\}, \\
    \regiononec &:= \{a_i \in \actionset_i : \actionmin_i < a_i^{+k} < \actionmax_i\}, \\
    \regiontwoa &:= \{a_i \in \actionset_i : a_i^{+k} \geq \actionmax_i, \ a_i > \actionmax_i - \constant_2 \alpha\}, \\
    \regiontwob &:= \{a_i \in \actionset_i : a_i^{+k} \leq \actionmin_i, \ a_i < \actionmin_i + \constant_2 \alpha\}. 
\end{align*}
In words, $\regiononea$ (resp. $\regiontwoa$) represents the region within $\actionset_i$ where $a_i^{+k}$ is on or outside the \emph{right}-hand-side boundary of $\actionset_i$, while $a_i$ is at least (resp. less than) $\constant_2 \alpha$ inside the boundary; $\regiononeb$ (resp. $\regiontwob$) represents the region within $\actionset_i$ where $a_i^{+k}$ is on or outside the \emph{left}-hand-side boundary of $\actionset_i$, while $a_i$ is at least (resp. less than) $\constant_2 \alpha$ inside the boundary; finally, $\regiononec$ represents the region within $\actionset_i$ where $a_i^{+k}$ is strictly inside the boundaries of $\actionset_i$. 
Note that these regions form a partition of $\actionset_i$. We refer to $\regionone := \regiononea \cup \regiononeb \cup \regiononec$ as the \emph{interior region} and $\regiontwo := \regiontwoa \cup \regiontwob$ as the \emph{exterior region} of the action set $\actionset_i$. 
In \textbf{Part 1} below, we show that if player~$i$ updates his action at time $k$ and $a^k_i \in \regionone$ (the interior region), then $\potentialincrement^k_i \geq \stepsizemin \constant_1^2 \alpha^2$ (Case 1), while if $a^k_i \in \regiontwo$ (the exterior region), then $\potentialincrement^k_i \geq 0$ (Case 2).  
Based on these facts, we can distinguish between the following two only possible scenarios: (S1) At time $k$, there exists a player~$i$ who updates his action, and $a^k_i \in \regionone$. In this scenario, using \eqref{eq: potential function increment L smooth property}, we obtain that $\potentialfunction(\actionvector^{\timeindex+1}) \geq \potentialfunction(\actionvector^\timeindex) + \stepsizemin \constant_1^2 \alpha^2$. (S2) At time $k$, either no player updates his action, or for every player~$i$ that updates his action, $a^k_i \in \regiontwo$. In this scenario, using \eqref{eq: potential function increment L smooth property}, we obtain that $\potentialfunction(\actionvector^{\timeindex+1}) \geq \potentialfunction(\actionvector^\timeindex)$. This ensures that $\{\potentialfunction(\actionvector^k)\}_{k = 0}^\infty$ forms a non-decreasing sequence which increases by at least $\stepsizemin \constant_1^2 \alpha^2 > 0$ whenever (S1) happens. Since $\potentialfunction$ is bounded from above, (S1) can happen only finitely many times. Hence, $\exists \bar{k} \geq 0$ such that for all $k \geq \bar{k}$, only (S2) can happen. Finally, in \textbf{Part 2} below, we show that if (S2) happens at time $k$, then $\actionvector^k$ is a $2\alpha$-Nash equilibrium of the game. 

For ease of notation, let $\gradientapprox^k_i := \partial/(\partial a_i)\utility_\playeridx(\actionvector^{\timeindex})$, $\gradientpotential^k_i := (\partial/ \partial a_i) \potentialfunction(\actionvector^k)$ for all $k \geq 0$, for all $i \in [\numplayers]$. Note that $\potentialincrement^k_i = \gradientpotential^k_i(a^{k+1}_i - a^k_i) - (L/2)(a^{k+1}_i - a^k_i)^2$.

\vspace{1ex}
\noindent\textbf{Part 1}: 
\emph{Case 1: Interior region:} Suppose player~$i$ updates his action at time $k$ such that $a^k_i \in \regionone$. We need to show that 
\begin{align}
\label{eq: potential function increment player i positive}
    \potentialincrement^k_i := \gradientpotential^k_i(a^{k+1}_i - a^k_i) -  \frac{L}{2} (a^{k+1}_i - a^k_i)^2 \geq \stepsizemin \constant_1^2 \alpha^2. 
\end{align}
We consider the following sub-cases. 

\noindent\emph{Case 1A: $a^k_i \in \regiononea$: } 
Using the relationship between the gradient of the utility function of player~$i$ and that of the $\alpha$-potential function in \eqref{eq: phi and u_i derivative relation}, we obtain that
\begin{align*}
    \gradientpotential^k_i \geq \gradientapprox^k_i - \frac{\alpha}{\adelta} = \gradientapprox^k_i - \constant_1 \alpha \geq \frac{\actionmax_i - a^k_i}{\stepsize_i} - \constant_1 \alpha,
\end{align*}
where we use the fact that $\constant_1 = 1/\adelta$ by definition, and $\action_\playeridx^\timeindex+\stepsize_i\gradientapprox^k_i \geq \actionmax_i$ since $a^k_i \in \regiononea$. Due to the projection operation in the update step in Algorithm \ref{alg: gradient play}, we have $a^{k+1}_i := \Projection{\actionset_i}{\action_\playeridx^\timeindex+\stepsize_i\gradientapprox^k_i} = \actionmax_i$. Hence, $\action_\playeridx^{\timeindex+1} - a^k_i \geq 0$. Thus, multiplying both sides of the inequality above by $\action_\playeridx^{\timeindex+1} - a^k_i$, we obtain that
\begin{align}
\label{eq: potential function increment lowerbound 1}
    \gradientpotential^k_i (\action_\playeridx^{\timeindex+1} - a^k_i) \geq \frac{(a^{k+1}_i - a^k_i)^2}{\stepsize_i} - \constant_1 \alpha (\action_\playeridx^{\timeindex+1} - a^k_i). 
\end{align}
Given $\stepsize_i \leq L \constant_2^2 \alpha^2 / (2 (\constant_1 \alpha \gradientbound + \constant_1^2 \alpha^2))$, by rearranging the terms, we obtain that 
\begin{align*}
    \frac{L}{2} \constant_2^2 \alpha^2 \geq \stepsize_i (\constant_1 \alpha \gradientbound + \constant_1^2 \alpha^2). 
\end{align*}
Further, since $\stepsize_i \leq 1/L$ by assumption, we have $1/\stepsize_i - \lipschitz/2 \geq \lipschitz/2 > 0$. Combining this with the inequality above, we obtain that
\begin{align*}
    \constant_2^2 \alpha^2 \left(\frac{1}{\stepsize_i} - \frac{\lipschitz}{2}\right) \geq \constant_2^2 \alpha^2 \frac{L}{2} \geq \constant_1 \alpha \stepsize_i \gradientbound + \stepsize_i \constant_1^2 \alpha^2. 
\end{align*}
Note that $1/\stepsize_i - \lipschitz/2 > 0$ and since $a^k_i \in \regiononea$, we have $a_i^{k+1} - a^k_i = \actionmax_i - a^k_i \geq \constant_2 \alpha \geq 0$ and $a_i^{k+1} - a^k_i = \actionmax_i - a^k_i \leq \stepsize_i \gradientapprox^k_i \leq \stepsize_i \gradientbound$ by Assumption~\ref{assn: continuity and compactness}. Hence, from the inequality above, we obtain that
\begin{align*}
    (a^{k+1}_i - a^k_i)^2 \left(\frac{1}{\stepsize_i} - \frac{\lipschitz}{2}\right) \geq \constant_1 \alpha (a^{k+1}_i - a^k_i) + \stepsize_i \constant_1^2 \alpha^2,
\end{align*}
which can be rewritten as
\begin{align*}
    \frac{(a^{k+1}_i - a^k_i)^2}{\stepsize_i} - \constant_1 \alpha (a^{k+1}_i - a^k_i) \geq \frac{L}{2} (a^{k+1}_i - a^k_i)^2 + \stepsize_i \constant_1^2 \alpha^2. 
\end{align*}
Combining the inequality above with \eqref{eq: potential function increment lowerbound 1} gives the required bound \eqref{eq: potential function increment player i positive} by noting that $\stepsize_i \geq \stepsizemin := \min_{i \in [\numplayers] \stepsize_i}$.

\noindent\emph{Case 1B: $a^k_i \in \regiononeb$: } 
The required result follows from similar steps as in Case 1A and is thus omitted.

\noindent\emph{Case 1C: $a^k_i \in \regiononec$: } 
In this case, $\action_\playeridx^{k+1} = \Projection{\actionset_i}{\action_\playeridx^\timeindex+\stepsize_i\gradientapprox^k_i} = \action_\playeridx^k + \stepsize_i \gradientapprox^k_i$. Hence, 
\begin{align}
\label{eq: L/2 update difference square upperbound}
      \frac{\lipschitz}{2} (a^{\timeindex+1}_\playeridx - a^k_i)^2 =  \frac{\lipschitz}{2}\stepsize_i^2  (\gradientapprox^k_i)^2 \leq \frac{\lipschitz}{2} \stepsize_i^2 \gradientbound^2
\end{align}
due to Assumption~\ref{assn: continuity and compactness}. 
To lower-bound $\gradientpotential^k_i (\action_\playeridx^{\timeindex+1} - a^k_i)$, we first write 
\begin{align}
\label{eq: grad phi_i times update difference}
    \gradientpotential^k_i (\action_\playeridx^{\timeindex+1} - a^k_i)  =\stepsize_i\gradientapprox^k_i\gradientpotential^k_i.  
\end{align}
Since player~$i$ updates at time $k$, from the condition given in Algorithm \ref{alg: gradient play}, it must be that $\left|\gradientapprox^k_i\right| > 2\alpha/\adelta = 2 \constant_1 \alpha$. Thus, either $\gradientapprox^k_i < -2 \constant_1 \alpha$ or $\gradientapprox^k_i > 2 \constant_1 \alpha$. 
If $\gradientapprox^k_i < -2\constant_1 \alpha$, then, using the relationship in \eqref{eq: phi and u_i derivative relation}, we obtain that
\begin{align*}
     \gradientpotential^k_i &\leq  \gradientapprox^k_i  +  \constant_1 \alpha <  - \constant_1 \alpha. 
\end{align*}
Else, $\gradientapprox^k_i > 2\constant_1 \alpha$ and hence
\begin{align*}
     \gradientpotential^k_i &\geq  \gradientapprox^k_i  -  \constant_1 \alpha >   \constant_1 \alpha. 
\end{align*}
Thus, in each case, it holds that
\begin{align*}
    \gradientapprox^k_i \gradientpotential^k_i  > 2\constant_1^2\alpha^2. 
\end{align*}
Hence, since $\stepsize_i \leq 2\constant_1^2\alpha^2/(\lipschitz\gradientbound^2)$ by assumption, we have 
\begin{align*}
    \stepsize_i \gradientapprox^k_i \gradientpotential^k_i > 2\stepsize_i \constant_1^2\alpha^2 \geq \frac{\lipschitz}{2} \stepsize_i^2 \gradientbound^2 + \stepsize_i \constant_1^2\alpha^2. 
\end{align*}
Combining the inequality above with \eqref{eq: grad phi_i times update difference} and \eqref{eq: L/2 update difference square upperbound}, we obtain the required bound \eqref{eq: potential function increment player i positive}  by noting that $\stepsize_i \geq \stepsizemin$.

\noindent \emph{Case 2: Exterior region:} Suppose player~$i$ updates his action at time $k$ such that $a^k_i \in \regiontwo$. We need to show that
\begin{align}
\label{eq: potential function increment player i non negative}
    \potentialincrement^k_i := \gradientpotential^k_i(a^{k+1}_i - a^k_i) - \frac{L}{2} (a^{k+1}_i - a^k_i)^2 \geq 0. 
\end{align}
We consider the following sub-cases. 

\noindent\emph{Case 2A: $a^k_i \in \regiontwoa$: }
Note that $a^k_i \leq \actionmax_i$ and since $a^k_i \in \regiontwoa$, we have $\action_\playeridx^\timeindex+\stepsize_i\gradientapprox^k_i \geq \actionmax_i$. Hence, it must be that $\gradientapprox^k_i \geq 0$. Further, since player~$i$ updates at time $k$, it must be that $\gradientapprox^k_i > 2\constant_1\alpha$. Using \eqref{eq: phi and u_i derivative relation}, this implies that $\gradientpotential^k_i > \constant_1\alpha$. Since $\actionmax_i - a^k_i \geq 0$, we obtain that
\begin{align*}
    \gradientpotential^k_i (\actionmax_i - a^k_i) \geq \constant_1\alpha (\actionmax_i - a^k_i). 
\end{align*}
Since $a^k_i \in \regiontwoa$, we have $a^{k+1}_i := \Projection{\actionset_i}{\action_\playeridx^\timeindex+\stepsize_i\gradientapprox^k_i} = \actionmax_i$ and $\actionmax_i - a^k_i < \constant_2 \alpha$. Further, $\constant_2 \leq 2 \constant_1 / L$ by definition. Hence, we have $(\lipschitz / 2)(\actionmax_i - a^k_i) < (L/2) \constant_2 \alpha \leq \constant_1 \alpha$. Using these facts with the inequality above, we obtain that
\begin{align*}
    \gradientpotential^k_i (a^{k+1}_i - a^k_i) &\geq \constant_1\alpha (a^{k+1}_i - a^k_i) \geq \frac{\lipschitz}{2} (a^{k+1}_i - a^k_i)^2
\end{align*}
as required in \eqref{eq: potential function increment player i non negative}.

\noindent\emph{Case 2B: $a^k_i \in \regiontwob$: } 
The required result follows from similar steps as in Case 2A and is thus omitted.

\vspace{1ex}
\noindent\textbf{Part 2:} We need to show that if (S2) happens at time $k$, i.e., either no player updates his action, or for every player~$i$ that updates his action, $a^k_i \in \regiontwo$, then $\actionvector^k$ is a $2\alpha$-Nash equilibrium. 
Suppose player~$i$ does not update his action at time $k$. By the first-order condition for concavity of $\utility_\playeridx(\cdot, \actionvector^k_{-i})$, for all actions $a_i \in \actionset_i$, it holds that
\begin{align*}
    \utility_\playeridx(a_i, \actionvector^k_{-i}) &\leq \utility_\playeridx(\actionvector^k) + \gradientapprox^k_i \cdot (a_i - a^k_i) \nonumber \\
    &\leq \utility_\playeridx(\actionvector^k) + 2 \constant_1 \alpha \cdot \adelta \nonumber \\
    &= \utility_\playeridx(\actionvector^k) + 2 \alpha,  
\end{align*}
where we use the fact that $|a_i - a^k_i| \leq \adelta$ by Assumption~\ref{assn: continuity and compactness}, $|\gradientapprox^k_i| \leq 2 \constant_1 \alpha$ and $\constant_1 = 1/\adelta$ by definition. 
Next, suppose player~$i$ updates his action and $a^k_i \in \regiontwoa$. Since $a^k_i \in \regiontwoa$, we have $a_i - a^k_i \leq \actionmax_i - a^k_i < \constant_2 \alpha$ for all $a_i \in \actionset_i$. Further, $\gradientapprox^k_i \geq 0$ (see Case 2A in Part 1) and $\gradientapprox^k_i \leq \gradientbound$ by Assumption~\ref{assn: continuity and compactness}. 
Thus, we have
\begin{align*}
    \utility_\playeridx(a_i, \actionvector^k_{-i}) &\leq \utility_\playeridx(\actionvector^k) + \gradientapprox^k_i \cdot (a_i - a^k_i) \nonumber \\
    &\leq \utility_\playeridx(\actionvector^k) + \gradientbound \cdot \constant_2 \alpha \nonumber \\
    &\leq \utility_\playeridx(\actionvector^k) + \alpha \\
    &\leq \utility_\playeridx(\actionvector^k) + 2\alpha,  
\end{align*}
where we use the fact that $\constant_2 \leq 1/\gradientbound$ by definition. 
Similar arguments apply for when $a^k_i \in \regiontwob$. 
Thus, if (S2) happens, then $\actionvector^k$ is a $2 \alpha$-Nash equilibrium of the game.

\section{Proof of Proposition~\ref{prop: alpha potential expression LQ}}
\label{sec: proof of alpha potential expression LQ}

Consider any player $i \in [\numplayers]$ and actions $a_i, a_i' \in \actionset_i, \actionvector_{-i} \in \actionset_{-i}$. Then, 
\begin{align*}
    &\potentialfunction(\actionvector) - \potentialfunction(a_i',\actionvector_{-i}) = -\frac{1}{2} (a_i^2 - a_i'^2) + \beta_i (a_i - a_i') + \frac{\gamma}{2} \sum_{j = 1}^\numplayers \left[G_{ij} a_j (a_i - a_i') + G_{ji} a_j (a_i - a_i')\right].  
\end{align*}
Further, 
\begin{align*}
    \utility_\playeridx(\actionvector) - \utility_\playeridx(a_i',\actionvector_{-i}) = -\frac{1}{2} (a_i^2 - a_i'^2) + \beta_i (a_i - a_i') + \gamma \sum_{j = 1}^\numplayers G_{ij} a_j (a_i - a_i').  
\end{align*}
Hence, 
\begin{align*}
    \left|\left[\utility_\playeridx(\actionvector) - \utility_\playeridx(a_i',\actionvector_{-i})\right] - \left[\potentialfunction(\actionvector) - \potentialfunction(a_i',\actionvector_{-i})\right]\right| &= \frac{|\gamma|}{2} \left|\sum_{j = 1}^\numplayers \left(G_{ij} a_j (a_i - a_i') - G_{ji} a_j (a_i - a_i')\right)\right| \\
    &\leq \frac{|\gamma|}{2} \sum_{j = 1}^\numplayers \left|G_{ij} - G_{ji}\right| \cdot |a_j| \cdot |a_i - a_i'| \\
    &\leq \frac{\abar \adelta |\gamma|}{2} \sum_{j = 1}^\numplayers \left|G_{ij} - G_{ji}\right|. 
\end{align*}
The result then follows from the definition of an $\alpha$-potential function (Definition~\ref{def: alpha potential function}).

\section{Proofs of bounds on \texorpdfstring{$\|\adjacency - \adjacency^T\|_\infty$}{norm of G - G transpose} for various networks in Section \ref{sec: alpha for various networks}}
\label{sec: proofs of alpha and G norm bounds}

To derive high probability bounds, we use the following result which is a simple consequence of Markov's inequality. 

\begin{lemma}{(Application of Markov's inequality)}
\label{lem: Markov inequality}
    Consider $\numplayers$ non-negative random variables $X_1, \dots, X_N$ each with expectation at most $e(\numplayers)$ such that $\numplayers e(\numplayers) \to 0$ as $\numplayers \to \infty$. 
    Then, given $t > 0$ and $\delta > 0$, for all $\numplayers$ large enough, with probability $1 - \delta$, we have $\max_{i \in [\numplayers]} X_i \leq e(\numplayers) + t$. 
\end{lemma}
\begin{proof}
    By Markov's inequality, for all fixed $t > 0$, 
    \begin{align*}
        \P(X_i \geq e(\numplayers) + t) \leq \frac{\E(X_i)}{e(\numplayers) + t} \leq \frac{e(\numplayers)}{e(\numplayers) + t}. 
    \end{align*}
    By the union bound over all events $X_i \geq e(\numplayers) + t$, we obtain that
    \begin{align*}
        \P\Big(\max_{i \in [\numplayers]} X_i \geq e(\numplayers) + t\Big) \leq \frac{Ne(\numplayers)}{e(\numplayers) + t} \leq \frac{Ne(\numplayers)}{t}. 
    \end{align*}
    Since $\numplayers e(\numplayers) \to 0$ as $\numplayers \to \infty$, given $\delta > 0$, for all $\numplayers$ large enough, 
    \begin{align*}
        \P\Big(\max_{i \in [\numplayers]} X_i \geq e(\numplayers) + t\Big) \leq \delta. 
    \end{align*}
\end{proof}

\subsection{Complete network with errors} 

First, we obtain $|G_{ij} - G_{ji}| \leq 2 \epsilon(\numplayers)$. Thus, $\alpha = \|\adjacency - \adjacency^T\|_\infty \leq 2 \numplayers \epsilon(\numplayers)$. 
Moreover, 
\begin{align*}
    \|\adjacency\|_2^2 &= \sup_{\actionvector \neq 0} \|\adjacency \actionvector\|_2^2/ \|\actionvector\|_2^2 \geq \|\adjacency \one{\numplayers}\|_2^2/ \|\one{\numplayers}\|_2^2 = \one{\numplayers}^T \adjacency^T \adjacency \one{\numplayers}/(\one{\numplayers}^T \one{\numplayers}) = \sum_{i, j} [\adjacency^T \adjacency]_{ij} / \numplayers \\
    &\geq \sum_{i, j} (\numplayers-2) (1 - \epsilon(\numplayers))^2 / \numplayers = \numplayers(\numplayers-2) (1 - \epsilon(\numplayers))^2. 
\end{align*}
Thus, $\|\adjacency\|_2 \geq \sqrt{\numplayers(\numplayers-2)} (1 - \epsilon(\numplayers))$. 
Also, 
\begin{align*}
    \|\adjacency\|_\infty &= \max_i \sum_{j = 1}^\numplayers |G_{ij}| \geq \max_i \sum_{j \neq i} (1 - \epsilon(\numplayers)) = (\numplayers-1) (1 - \epsilon(\numplayers)). 
\end{align*}

\subsection{Complete network with an influential player} 

The bounds on $\|\adjacency - \adjacency^T\|_\infty$ and $\|\adjacency\|_2$ follow by the same arguments as above. Further,
\begin{align*}
    \|\adjacency\|_\infty &= \max_i \sum_{j = 1}^\numplayers |G_{ij}| \geq \sum_{j = 1}^\numplayers |G_{Nj}| \geq \sum_{j \neq i} (w - \epsilon(\numplayers)) = (\numplayers-1) (w - \epsilon(\numplayers)). 
\end{align*}

\subsection{Complete network with random signs} 

Fix any $i < j$. With probability $1 - \delta(\numplayers)$, the sign of $G_{ij}$ is not flipped. Thus, we obtain $|G_{ij} - G_{ji}| \leq 2 \epsilon(\numplayers)$ with probability $1 - \delta(\numplayers)$ as in the case of a complete network with errors (without signs). With probability $\delta(\numplayers)$, we obtain $|G_{ij} - G_{ji}| \leq 2 (1 + \epsilon(\numplayers))$ due to the sign of $G_{ij}$ and $G_{ji}$ being opposite. Thus, $\E(|G_{ij} - G_{ji}|) \leq 2 \epsilon(\numplayers) (1 - \delta(\numplayers)) + 2 (1 + \epsilon(\numplayers)) \delta(\numplayers) = 2 (\epsilon(\numplayers) + \delta(\numplayers))$. 

Consider the case $\epsilon(\numplayers) = \delta(\numplayers) = 1/\numplayers^r$ with $r > 2$. Let $\sum_{j = 1}^\numplayers |G_{ij} - G_{ji}| =: X_i$. Then, for all $i \in [\numplayers]$, $\E(X_i) \leq 2 \numplayers (\epsilon(\numplayers) + \delta(\numplayers)) = 4 / \numplayers^{r - 1}$. Note that $\numplayers \E(X_i) = 4 / \numplayers^{r - 2} \to 0$ as $\numplayers \to \infty$. Thus, by Lemma \ref{lem: Markov inequality}, we obtain that, for any $\delta > 0$ and $t > 0$, with probability $1 - \delta$, $\alpha = \|\adjacency - \adjacency^T\|_\infty = \max_{i \in [\numplayers]} X_i \leq 4/\numplayers^{r-1} + t$.

\subsection{\erdosrenyi network}

Given any $i \neq j$, $\E(|G_{ij} - G_{ji}|) = 2 p(\numplayers)(1 - p(\numplayers))$. 
Given that either $p(\numplayers) \leq 1/\numplayers^r$ or $p(\numplayers) \geq 1 - 1/\numplayers^r$, we have $p(\numplayers) (1 - p(\numplayers)) \leq 1/\numplayers^r$. Suppose $r > 2$. Let $\sum_{j = 1}^\numplayers |G_{ij} - G_{ji}| =: X_i$. Then, for all $i \in [\numplayers]$, $\E(X_i) \leq 2 \numplayers p(\numplayers)(1 - p(\numplayers)) \leq 2 / \numplayers^{r-1}$. Note that $\numplayers \E(X_i) \leq 2/\numplayers^{r - 2}$. Thus, by Lemma \ref{lem: Markov inequality}, we obtain that, for any $\delta > 0$ and $t > 0$, with probability $1 - \delta$, $\alpha = \|\adjacency - \adjacency^T\|_\infty = \max_{i \in [\numplayers]} X_i \leq 2 / \numplayers^{r-1} + t$.

\subsection{Small world network}

Fix any $i \neq j$. We consider the following two cases. (a)
Suppose that before rewiring the network, edge $(j,i)$ exists. 
Then, $G_{ij} = 0$ if and only if edge $(j,i)$ is selected for rewiring and connected to $k \neq i$. Thus, $\P(G_{ij} = 0) = p(\numplayers) \cdot (\numplayers-2)/(\numplayers-1) \leq p(\numplayers)$. Note that the two events $G_{ij} = 0$ and $G_{ji} = 1$ are independent. Hence, $\P(G_{ij} = 0, G_{ji} = 1) = \P(G_{ij} = 0) \cdot \P(G_{ji} = 1) \leq p(\numplayers)$. 
(b) Next, consider the case where before rewiring the network, the edge $(j,i)$ does not exist. Then, $G_{ij} = 1$ if and only if for some $k$ such that edge $(j,k)$ exists before rewiring the network (note that there are $2d(\numplayers)$ such edges), the edge is chosen for rewiring and connected to $i$. Note that each of these $2d(\numplayers)$ events has probability $p(\numplayers) \cdot 1/(\numplayers-1)$. By the union bound, the probability that at least one of them happens satisfies $\P(G_{ij} = 1) \leq 2 d(\numplayers) \cdot p(\numplayers) \cdot 1 / (\numplayers-1)$. Hence, $\P(G_{ij} = 1, G_{ji} = 0) = \P(G_{ij} = 1) \cdot \P(G_{ji} = 0) \leq 2 d(\numplayers) p(\numplayers) / (\numplayers-1)$. 

Fix any node $i \in [\numplayers]$. It has $2 d(\numplayers)$ neighbours $j$ such that the edges $(i,j)$ and $(j,i)$ exist before rewiring. Thus, using the bounds derived above, 
\begin{align*}
    \E\left(\sum_{j = 1}^\numplayers |G_{ij} - G_{ji}|\right) &= \sum_{j = 1}^\numplayers \left( \P(G_{ij} = 0, G_{ji} = 1) + \P(G_{ij} = 1, G_{ji} = 0)\right) \\
    &= 2 \sum_{j = 1}^\numplayers \P(G_{ij} = 0, G_{ji} = 1) \\
    &= 2 \left( \sum_{\substack{j : (i,j) \text{ exists } \\ \text{ before rewiring}}} \P(G_{ij} = 0, G_{ji} = 1) + \sum_{\substack{j : (i,j) \text{ does not exist } \\ \text{ before rewiring}}} \P(G_{ij} = 0, G_{ji} = 1)\right) \\
    &\leq 2 \left( 2 d(\numplayers) \cdot p(\numplayers) + \underbrace{(\numplayers - 2 d(\numplayers))}_{\leq \numplayers-1} \cdot \frac{2 d(\numplayers) p(\numplayers)}{\numplayers-1}\right) \\
    &\leq 8 d(\numplayers) p(\numplayers). 
\end{align*}
Let $X_i := \sum_{j = 1}^\numplayers |G_{ij} - G_{ji}|$. Then, for all $i \in [\numplayers]$, $\E(X_i) \leq 8 d(\numplayers) p(\numplayers)$. Suppose $d(\numplayers) = \numplayers / 8$, $p(\numplayers) = 1/\numplayers^r$ with $r > 2$. Hence, $\numplayers \E(X_i) = 1/\numplayers^{r-2} \to 0$ as $\numplayers \to \infty$. Thus, by Lemma \ref{lem: Markov inequality}, we obtain that, for any $\delta > 0$ and $t > 0$, with probability $1 - \delta$, $\alpha = \|\adjacency - \adjacency^T\|_\infty = \max_{i \in [\numplayers]} X_i \leq 1/\numplayers^{r-1} + t$.

\subsection{Star network with erased edges}

For $i > 1, j > 1$, there does not exist an edge between nodes $i, j$. Hence, for such $i, j$, we have $|G_{ij} - G_{ji}| = 0$. 
For $j > 1$, there exists an edge from node $1$ to node $j$ with probability $1 - p(\numplayers)$ and from node $j$ to node $1$ independently with probability $1 - p(\numplayers)$. Hence, for all $j > 1$, $|G_{1j} - G_{j1}| = 1$ with probability $2 p(\numplayers) (1 - p(\numplayers))$ and $|G_{1j} - G_{j1}| = 0$ otherwise. 
Thus, for $i > 1$, $\E(\sum_{j = 1}^\numplayers |G_{ij} - G_{ji}|) = \E(|G_{i1} - G_{1i}|) = 2 p(\numplayers) (1 - p(\numplayers)) \leq 2 p(\numplayers)$. 
Further, for $i = 1$, $\E(\sum_{j = 1}^\numplayers |G_{1j} - G_{j1}|) \leq 2 (\numplayers-1) p(\numplayers) (1 - p(\numplayers)) \leq 2 \numplayers p(\numplayers)$.
Thus, for all $i \in [\numplayers]$, $\E(\sum_{j = 1}^\numplayers |G_{ij} - G_{ji}|) \leq 2 \numplayers p(\numplayers)$. 

Let $\sum_{j = 1}^\numplayers |G_{ij} - G_{ji}| =: X_i$. Suppose $p(\numplayers) = 1/\numplayers^r$ with $r > 2$. Then, $\E(X_i) \leq 2 \numplayers p(\numplayers) = 2/\numplayers^{r-1}$. Hence, $\numplayers \E(X_i) \leq 2 / \numplayers^{r-2} \to 0$ as $\numplayers \to \infty$. Thus, by Lemma \ref{lem: Markov inequality}, we obtain that, for any $\delta > 0$ and $t > 0$, with probability $1 - \delta$, $\alpha = \|\adjacency - \adjacency^T\|_\infty = \max_{i \in [\numplayers]} X_i \leq 2/\numplayers^{r-1} + t$.

\section{Proof of Theorem~\ref{thm: pos}}
\label{sec: proof of theorem pos}

By definition of $\sw$ in \eqref{eq: social welfare} and Lemma \ref{lem: unique social optimum} in Appendix~\ref{sec: appendix intermediate results}, we obtain that
{\allowdisplaybreaks
\begin{align*}
    \sw(\actionvectoropt) &= -\frac{1}{2} (\actionvectoropt)^T \actionvectoropt + (\actionvectoropt)^T \left( \betavector + \gamma \adjacency \actionvectoropt\right) \\
    &= -\frac{1}{2} (\actionvectoropt)^T \actionvectoropt + (\actionvectoropt)^T \left( \betavector + \gamma (\adjacency + \adjacency^T) \actionvectoropt - \gamma \adjacency^T \actionvectoropt\right) \\
    &= -\frac{1}{2} (\actionvectoropt)^T \actionvectoropt + (\actionvectoropt)^T \left(\actionvectoropt - \gamma \adjacency^T \actionvectoropt\right) \\
    &= \frac{1}{2} (\actionvectoropt)^T \actionvectoropt - \gamma (\actionvectoropt)^T \adjacency^T \actionvectoropt \\
    &= (\actionvectoropt)^T \left(\frac{\identity}{2} - \gamma \adjacency^T\right) \actionvectoropt \\
    &\stackrel{(a)}{=} (\actionvectoropt)^T \left(\frac{\identity}{2} - \frac{\gamma(\adjacency + \adjacency^T)}{2}\right) \actionvectoropt \\
    &\stackrel{(b)}{\leq} \lambdamax{\frac{\identity}{2} - \gamma\frac{\adjacency + \adjacency^T}{2}} \|\actionvectoropt\|_2^2 \\
    &= \frac{\lambdamax{\identity - \gamma(\adjacency + \adjacency^T)}}{2} \betavector^T \left(\identity - \gamma (\adjacency + \adjacency^T)\right)^{-2} \betavector \\
    &\stackrel{(c)}{\leq} \frac{\lambdamax{\identity - \gamma(\adjacency + \adjacency^T)}}{2} \lambdamax{\left(\identity - \gamma (\adjacency + \adjacency^T)\right)^{-2}} \|\betavector\|_2^2 \\
    &\stackrel{}{=} \frac{\lambdamax{\identity - \gamma(\adjacency + \adjacency^T)}}{2\lambdamin{\identity - \gamma (\adjacency + \adjacency^T)}^{2}} \|\betavector\|_2^2 \\
    &\stackrel{}{=} \frac{1 - \lambdamin{\gamma(\adjacency + \adjacency^T)}}{2(1 - \lambdamax{\gamma (\adjacency + \adjacency^T)}^{2}} \|\betavector\|_2^2 \\
    &\stackrel{(d)}{>} 0, 
\end{align*}
where $(a)$ follows from the fact that for any matrix $A$ and vector $x$, $x^T A x = x^T ((A + A^T)/2) x$, $(b)$ holds since for any symmetric matrix $A$, $x^T A x \leq \lambdamax{A} \|x\|_2^2$, $(c)$ holds since $\lambdamax{\identity - \gamma(\adjacency + \adjacency^T)} > 0$ by Lemma \ref{lem: lambda max G mat} and $(d)$ follows from Lemma \ref{lem: lambda max G mat}. 

Similarly, by Lemma \ref{lem: unique alpha potential maximum} in Appendix~\ref{sec: appendix intermediate results}, we obtain that 
\begin{align*}
    \sw(\actionvectorpotential) \nonumber &= -\frac{1}{2} (\actionvectorpotential)^T \actionvectorpotential + (\actionvectorpotential)^T \left( \betavector + \gamma \adjacency \actionvectorpotential\right) \\
    &= -\frac{1}{2} (\actionvectorpotential)^T \actionvectorpotential + (\actionvectorpotential)^T \left( \betavector + \frac{\gamma}{2} (\adjacency + \adjacency^T) \actionvectorpotential + \frac{\gamma}{2} (\adjacency - \adjacency^T) \actionvectorpotential\right) \\
    &= -\frac{1}{2} (\actionvectorpotential)^T \actionvectorpotential + (\actionvectorpotential)^T \left(\actionvectorpotential + \frac{\gamma}{2}(\adjacency - \adjacency^T) \actionvectorpotential\right) \\
    &= \frac{1}{2} (\actionvectorpotential)^T \actionvectorpotential + \frac{\gamma}{2} (\actionvectorpotential)^T (\adjacency - \adjacency^T) \actionvectorpotential \\
    &= \frac{1}{2} \|\actionvectorpotential\|_2^2 \\
    &= \frac{1}{2} \betavector^T \left(\identity - \frac{\gamma(\adjacency + \adjacency^T)}{2}\right)^{-2} \betavector \\
    &\geq \frac{1}{2} \lambdamin{\left(\identity - \frac{\gamma(\adjacency + \adjacency^T)}{2}\right)^{-2}} \|\betavector\|_2^2 \\
    &= \frac{1}{2\lambdamax{\identity - \frac{\gamma(\adjacency + \adjacency^T)}{2}}^2} \|\betavector\|_2^2 \\
    &= \frac{1}{2\left(1 - \frac{\lambdamin{\gamma(\adjacency + \adjacency^T)}}{2}\right)^2} \|\betavector\|_2^2 \\
    &> 0. 
\end{align*}
Thus, 
\begin{align*}
    \frac{\sw(\actionvectoropt)}{\sw(\actionvectorpotential)} &\leq \frac{(1 - \lambdamin{\gamma(\adjacency + \adjacency^T)} \left(1 - \frac{\lambdamin{\gamma(\adjacency + \adjacency^T)}}{2}\right)^2}{\left(1 - \lambdamax{\gamma(\adjacency + \adjacency^T)}\right)^2} \\
    &\leq \frac{\left(1 + \left|\lambdamin{\gamma(\adjacency + \adjacency^T)}\right|\right) \left(1 + \frac{\left|\lambdamin{\gamma(\adjacency + \adjacency^T)}\right|}{2}\right)^2}{\left(1 - \left|\lambdamax{\gamma(\adjacency + \adjacency^T)}\right|\right)^2} \\
    &\leq \frac{\left(1 + \left|\lambdamin{\gamma(\adjacency + \adjacency^T)}\right|\right)^3}{\left(1 - \left|\lambdamax{\gamma(\adjacency + \adjacency^T)}\right|\right)^2} \\
    &\leq \frac{\left(1 + |\gamma| \max_{i \in [\numplayers]} \left(\sum_{j \neq i} |G_{ij}|+\sum_{j \neq i} |G_{ji}|\right)\right)^3}{\left(1 - |\gamma| \max_{i \in [\numplayers]} \left(\sum_{j \neq i} |G_{ij}| + \sum_{j \neq i} |G_{ji}|\right)\right)^2}
\end{align*}
where the last inequality is due to the Gershgorin Disks Theorem \cite[Theorem~2.8]{bullo2018lectures}. Further, given $|G_{ij}| \leq 1/\numplayers$ for all $i, j$, then 
\begin{align*}
    \frac{\left(1 + |\gamma| \max_{i \in [\numplayers]} \left(\sum_{j \neq i} |G_{ij}|+\sum_{j \neq i} |G_{ji}|\right)\right)^3}{\left(1 - |\gamma| \max_{i \in [\numplayers]} \left(\sum_{j \neq i} |G_{ij}| + \sum_{j \neq i} |G_{ji}|\right)\right)^2} \leq \frac{\left(1+2|\gamma|\right)^3}{\left(1-2|\gamma|\right)^2}. 
\end{align*}}

\section{Intermediate results}
\label{sec: appendix intermediate results}

We recall the following well-known property of an $\lipschitz$-smooth function. 
\begin{lemma}{(Lower quadratic bound for $\lipschitz$-smooth functions \cite[Lemma~3.4]{bubeck2015convexopt})}
\label{lem: smooth function property}
    Let $\actionset \subseteq \R^\numplayers$ be a convex set and $\potentialfunction : \actionset \to \R$ be $\lipschitz$-smooth. Then, for all $\actionvector, \actionvector' \in \actionset$, it holds that
        \begin{align*}
            \potentialfunction(\actionvector) \geq \potentialfunction(\actionvector') + \nabla \potentialfunction(\actionvector')^T (\actionvector - \actionvector') - \frac{\lipschitz}{2} \|\actionvector - \actionvector'\|_2^2. 
        \end{align*}
\end{lemma}
The $\alpha$-potential function in \eqref{eq: alpha potential expression} has the following useful properties. 

\begin{lemma}{(Properties of $\potentialfunction$ in \eqref{eq: alpha potential expression})}
\label{lem: alpha potential properties}
    Consider a game satisfying Assumption~\ref{assn: continuity and compactness}, and assume $\utility_\playeridx \in \mathcal{C}^3(\R^\numplayers)$ for all $i \in [\numplayers]$. Then, the $\alpha$-potential function $\potentialfunction$ in \eqref{eq: alpha potential expression} satisfies the following properties. 
    \begin{enumerate}
        \item $\potentialfunction$ is bounded, i.e., $\exists \potentialfunction^{\max} \geq 0$ such that $\max_{\actionvector \in \actionset} |\potentialfunction(\actionvector)| \leq \potentialfunction^{\max}$. 
        \item $\potentialfunction$ is $\lipschitz$-smooth. 
    \end{enumerate}
\end{lemma}

\begin{proof}
Under Assumption~\ref{assn: continuity and compactness}: 
\begin{enumerate}
    \item $\potentialfunction : \actionset \to \R$ is a continuous function defined on a compact set and is thus bounded. 
    
    \item The derivative of $\potentialfunction$, i.e., $\nabla \potentialfunction : \actionset \to \R^\numplayers$ is well-defined and its expression is given by equation \eqref{eq: partial derivative potential function} in Appendix~\ref{sec: proof of theorem alpha potential existence}. Since $\utility_\playeridx \in \mathcal{C}^3(\R^\numplayers)$ for all $i \in [\numplayers]$, we obtain that $\nabla^2 \potentialfunction : \actionset \to \R^{\numplayers \times \numplayers}$ exists and is continuous. Since $\actionset$ is compact, we obtain that $\exists \lipschitz \geq 0$ such that $\max_{\actionvector \in \actionset} \|\nabla^2 \potentialfunction(\actionvector)\|_2 \leq \lipschitz$. Hence, using Taylor's theorem \cite[Theorem~2.1]{nocedal1999numerical}, we obtain that for all $\actionvector, \actionvector' \in \actionset$, 
    \begin{align*}
        \|\nabla \potentialfunction(\actionvector) - \nabla \potentialfunction(\actionvector')\|_2^2 &= \left\|\int_0^1 \nabla^2 \potentialfunction(\actionvector' + t(\actionvector - \actionvector')) \cdot (\actionvector - \actionvector') dt\right\|_2^2 \\
        &\leq \int_0^1 \|\nabla^2 \potentialfunction(\actionvector' + t(\actionvector - \actionvector'))\|_2^2 \cdot \|\actionvector - \actionvector'\|_2^2 dt \\
        &\leq \lipschitz^2 \|\actionvector - \actionvector'\|_2^2. 
    \end{align*}
\end{enumerate} 
\end{proof}

\begin{lemma}
\label{lem: lambda max G mat}
     Under Assumption~\ref{assn: contraction}, it holds that $\spectralradius{\gamma (\adjacency + \adjacency^T)} < 1$ and $\spectralradius{\gamma (\adjacency + \adjacency^T)/2} < 1/2$. Hence, the matrices $\identity - \gamma (\adjacency + \adjacency^T)$ and $\identity - \gamma (\adjacency + \adjacency^T)/2$ are invertible. 
\end{lemma}
\begin{proof}
    Let $\Gmat := \gamma (\adjacency + \adjacency^T)$.
    Note $\Gmat_{ii} = 0$ since $G_{ii} = 0$ for all $i \in [\numplayers]$. For $j \neq i$, 
    \begin{align*}
        |\Gmat_{ij}| \leq |\gamma| \cdot (|G_{ij}| + |G_{ji}|) \leq \frac{2 |\gamma|}{\numplayers} < \frac{1}{\numplayers}. 
    \end{align*}
    Hence, 
    \begin{align*}
        \sum_{j \neq i} |\Gmat_{ij}| < \frac{\numplayers - 1}{\numplayers} < 1. 
    \end{align*}
    Thus, by the Gershgorin Disks Theorem \cite[Theorem~2.8]{bullo2018lectures}, all eigenvalues of $\Gmat$ are contained within disks centered at $0$ with a radius less than $1$ in the complex plane. Thus, $\spectralradius{\Gmat} < 1$ and $\spectralradius{\Gmat/2} < 1/2$. This implies $\lambdamin{\identity - \Gmat} > 0$ and hence $\identity - \Gmat$ is invertible. 
    Similarly, $\lambdamin{\identity - \Gmat/2} > 1/2 > 0$ and hence $\identity - \Gmat/2$ is invertible. 
\end{proof}

\begin{lemma}
\label{lem: I - G mat inverse beta in A n}
    Under Assumption~\ref{assn: contraction}, it holds that $\inv{\identity - \gamma (\adjacency + \adjacency^T)} \betavector \in \actionset$, where $\actionset$ is the set of all action profiles. 
\end{lemma}
\begin{proof}
    Let $\Gmat := \gamma (\adjacency + \adjacency^T)$. 
    By Lemma \ref{lem: lambda max G mat}, we know that $\spectralradius{\Gmat} < 1$ and $\inv{\identity - \Gmat}$ exists. By the Neumann series expansion, we have
    \begin{align*}
        \inv{\identity - \Gmat} = \identity + \Gmat + \Gmat^2 + \dots. 
    \end{align*}
    Note that for all $i \neq j$, $|\Gmat_{ij}| \leq 2|\gamma|/\numplayers$ and for all $i$, $\Gmat_{ii} = 0$. Hence, 
    \begin{align*}
        \left|[\Gmat^2]_{ij}\right| = \left|\sum_{l = 1}^\numplayers \Gmat_{il} \Gmat_{lj}\right| \leq \sum_{l = 1}^\numplayers \frac{(2 |\gamma|)^2}{\numplayers^2} = \frac{(2 |\gamma|)^2}{\numplayers}. 
    \end{align*}
    By induction, we can show that for all $k \geq 1$, for all $i, j \in [\numplayers]^2$,
    \begin{align*}
        \left|[\Gmat^k]_{ij}\right| \leq \frac{(2|\gamma|)^k}{\numplayers}. 
    \end{align*}
    Let $\betabar = \max_{i \in [\numplayers]} |\beta_i|$. Thus, for all $i \in [\numplayers]$, 
    \begin{align*}
        \left|[\inv{\identity - \Gmat} \betavector]_i\right| &= \left|[(\identity + \Gmat + \Gmat^2 + \dots) \betavector]_i\right| \leq |\beta_i| + \sum_{j = 1}^\numplayers (|\Gmat_{ij}| + |[\Gmat^2]_{ij}| + \dots) |\beta_j| \\
        &\leq \betabar + \sum_{j = 1}^\numplayers \left(\frac{2|\gamma|}{\numplayers} + \frac{(2|\gamma|)^2}{\numplayers} + \dots\right) \betabar = \betabar\left(1 + 2|\gamma|+ (2|\gamma|)^2 + \dots\right) \\
        &= \frac{\betabar}{1 - 2 |\gamma|} \leq \actiontilde,  
    \end{align*}
    where the last inequality is due to Assumption~\ref{assn: contraction}. 
    Thus, for all $i \in [\numplayers]$, $[\inv{\identity - \Gmat} \betavector]_i \in [-\actiontilde, \actiontilde] \subseteq \actionset_i$. 
\end{proof}

\begin{lemma}
\label{lem: unique social optimum}
    Under Assumption~\ref{assn: contraction}, the unique maximizer of $\sw$ given in \eqref{eq: social welfare} is given by
    \begin{align*}
        \actionvectoropt = \inv{\identity - \gamma (\adjacency + \adjacency^T)}\betavector \in \actionset. 
    \end{align*}
\end{lemma}
\begin{proof}
    Recall that the social welfare function of an LQ network game is given by
    \begin{align*}
        \sw(\actionvector) = -\frac{1}{2} \actionvector^T \actionvector + \betavector^T \actionvector + \gamma \actionvector^T \adjacency \actionvector. 
    \end{align*}
    We first compute the maximizer of $\sw$ over the unconstrained domain $\R^\numplayers$. Note that
    \begin{align*}
        \nabla_\actionvector \sw(\actionvector) = -\actionvector + \betavector + \gamma (\adjacency + \adjacency^T) \actionvector 
    \end{align*}
    and 
    \begin{align*}
        \nabla^2_{\actionvector,\actionvector} \sw(\actionvector) = -\identity + \gamma (\adjacency + \adjacency^T). 
    \end{align*}
    By Lemma \ref{lem: lambda max G mat}, we know that $\lambdamax{-\identity + \gamma (\adjacency + \adjacency^T)} < 0$. Hence, $\sw$ is strongly concave. By the first-order optimality condition, the unique maximizer $\actionvectoropt$ of the unconstrained $\sw$ function must satisfy
    \begin{align*}
        \nabla_\actionvector \sw(\actionvectoropt) = -\actionvectoropt + \betavector + \gamma (\adjacency + \adjacency^T) \actionvectoropt = 0. 
    \end{align*}
    Hence, 
    \begin{align*}
        \actionvectoropt = \betavector + \gamma (\adjacency + \adjacency^T) \actionvectoropt 
    \end{align*}
    which implies
    \begin{align*}
        \left(\identity - \gamma (\adjacency + \adjacency^T)\right) \actionvectoropt = \betavector. 
    \end{align*}
    By Lemma \ref{lem: lambda max G mat}, the matrix above is invertible. Hence, 
    \begin{align*}
        \actionvectoropt = \inv{\identity - \gamma (\adjacency + \adjacency^T)}\betavector. 
    \end{align*}
    Further, by Lemma \ref{lem: I - G mat inverse beta in A n}, we have that $\actionvectoropt \in \actionset$. Thus, $\actionvectoropt$ is the unique maximizer of $\sw$. 
\end{proof}

\begin{lemma}
\label{lem: unique alpha potential maximum}
    Under Assumption~\ref{assn: contraction}, the unique maximizer of $\potentialfunction$ given in \eqref{eq: LQ alpha potential function} is given by
    \begin{align*}
        \actionvectorpotential = \inv{\identity - \frac{\gamma}{2} (\adjacency + \adjacency^T)} \betavector \in \actionset. 
    \end{align*}
\end{lemma}
The proof is similar to that of Lemma \ref{lem: unique social optimum}.

\section{Discussion}

\subsection{An alternative to Algorithm \ref{alg: gradient play}}
\label{sec: appendix alternative to gradient play algorithm}

An alternative to Algorithm~\ref{alg: gradient play} is the following algorithm. 
At each time $k \geq 0$, each player $i \in [\numplayers]$ updates his action as 
\begin{align}
\label{eq: gradient play new utilities}
    \action_{\playeridx}^{\timeindex+1} = \Projection{\actionset_i}{\action_{\playeridx}^{\timeindex} + \stepsize_\playeridx \frac{\partial}{\partial a_i}\potentialfunction(\actionvector^{\timeindex})}, 
\end{align} 
where $\potentialfunction$ is an $\alpha$-potential function of the game satisfying \eqref{eq: phi and u_i derivative relation}. 
This update rule can be justified as follows. Suppose a central planner knows the $\alpha$-potential function and the utility functions of the game. Then, to each player $i \in [\numplayers]$, the planner makes a payment of $e_i(\actionvector) := \potentialfunction(\actionvector) - \utility_\playeridx(\actionvector)$ for every action profile $\actionvector \in \actionset$. Thus, the new utility of player~$i$ is $v_\playeridx(\actionvector) := \utility_\playeridx(\actionvector) + e_i(\actionvector) = \potentialfunction(\actionvector)$. 
Under suitable regularity assumptions on $\potentialfunction$, such as smoothness and concavity, it can be guaranteed that (i) $\{\actionvector^k\}_{k = 0}^\infty$ converges 
to a maximizer of $\potentialfunction$, which is an $\alpha$-Nash equilibrium of the game, and (ii) the total payment made by the central planner at each iteration to each player~$i$ is $e_i(\actionvector)$ which is at most $\numplayers \alpha + \max_{i \in [\numplayers]} \inf_{\actionvectorfixed \in \actionset} (\potentialfunction(\actionvectorfixed) - \utility_\playeridx(\actionvectorfixed))$.\footnote{By the fundamental theorem of calculus, we obtain that for any $\actionvector, \actionvectorfixed \in \actionset$, $|\potentialfunction(\actionvector) - \utility_\playeridx(\actionvector) - (\potentialfunction(\actionvectorfixed) - \utility_\playeridx(\actionvectorfixed))| = \left|\int_0^1 [\nabla \potentialfunction(\actionvectorfixed + \epsilon(\actionvector - \actionvectorfixed)) - \nabla \utility_\playeridx(\actionvectorfixed + \epsilon(\actionvector - \actionvectorfixed))]^T \cdot (\actionvector - \actionvectorfixed) d\epsilon\right| \leq \numplayers \alpha$, by using the inequality \eqref{eq: phi and u_i derivative relation}. Hence, for all $\actionvectorfixed \in \actionset$, $e_i(\actionvector^k) = \potentialfunction(\actionvector^k) - \utility_\playeridx(\actionvector^k) = \potentialfunction(\actionvector^k) - \utility_\playeridx(\actionvector^k) - (\potentialfunction(\actionvectorfixed) - \utility_\playeridx(\actionvectorfixed)) + \potentialfunction(\actionvectorfixed) - \utility_\playeridx(\actionvectorfixed) \leq \numplayers \alpha + \potentialfunction(\actionvectorfixed) - \utility_\playeridx(\actionvectorfixed)$. } Thus, the update rule in \eqref{eq: gradient play new utilities} achieves convergence to a specific $\alpha$-Nash equilibrium corresponding to a maximizer of the $\alpha$-potential function (see Section~\ref{sec: price of stability}), while Algorithm~\ref{alg: gradient play} only ensures convergence to the set of $2\alpha$-Nash equilibria.

\subsection{Numerical example to illustrate Theorem \ref{thm: pos}}
\label{sec: appendix numerical example for pos}

We present an example to illustrate the bound in Theorem~\ref{thm: pos} numerically on a sequence of growing \erdosrenyi networks. 

\begin{example}
    Consider an LQ network game with $\gamma = 1/4$, and a sequence of network games sampled using the \erdosrenyi model as described in Example \ref{ex: dense erdos renyi networks} with edge probabilities $p(\numplayers) = 1/\numplayers$ and $p(\numplayers) = 1 - 1/\numplayers$, respectively. We set $G_{ij} = 1/\numplayers$ if there is an edge from $j$ to $i$ to satisfy Assumption~\ref{assn: contraction}. A plot of the welfare suboptimality bounds in Theorem~\ref{thm: pos} for such networks as a function of $\numplayers$ is shown in Figure \ref{fig: Pos ER graphs}. 
    \begin{figure}[t]
        \centering
        \begin{subfigure}[b]{0.49\textwidth}
            \includegraphics[width=\linewidth]{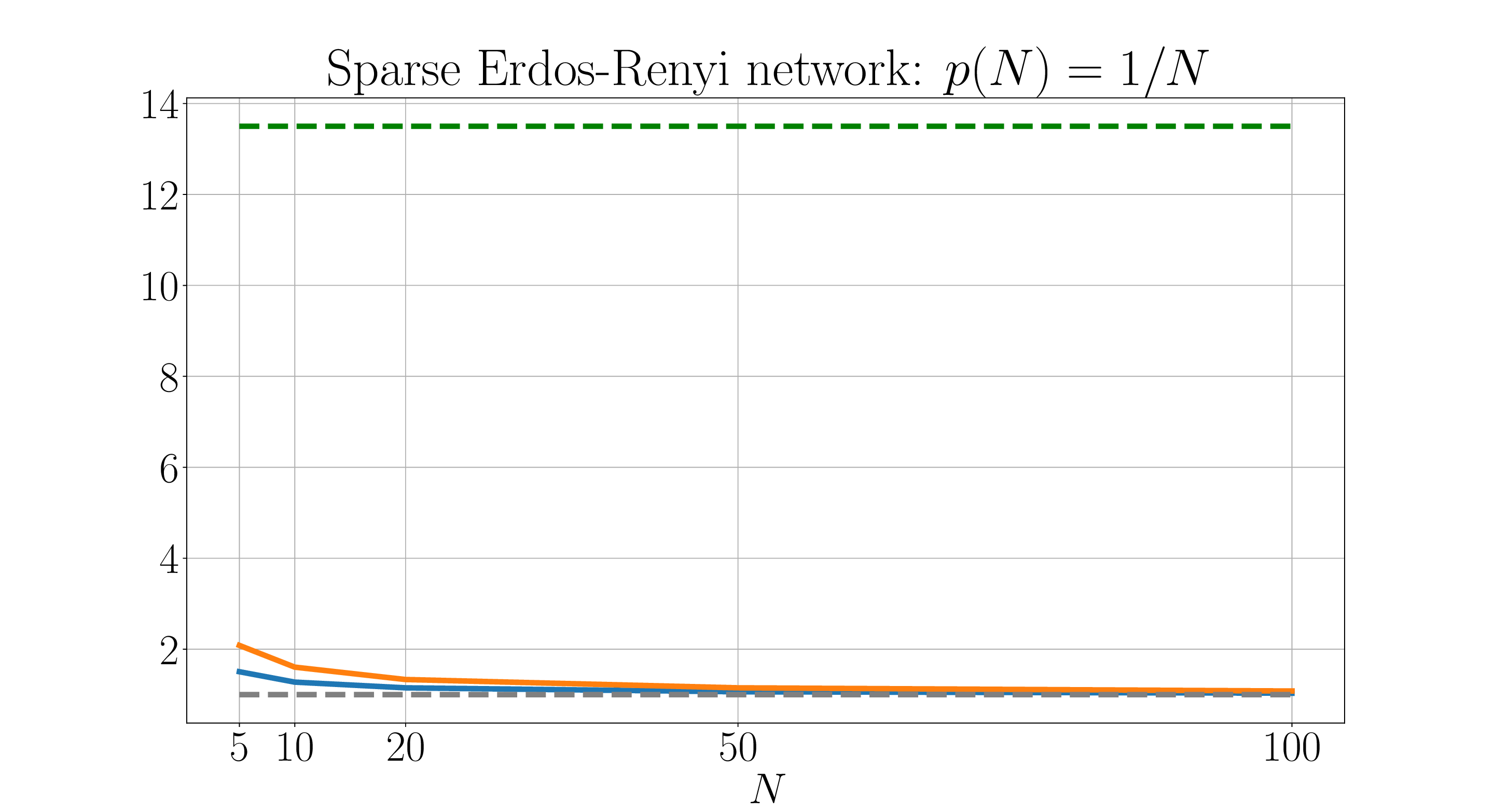}
        \end{subfigure}
        \begin{subfigure}[b]{0.49\textwidth}
            \includegraphics[width=\linewidth]{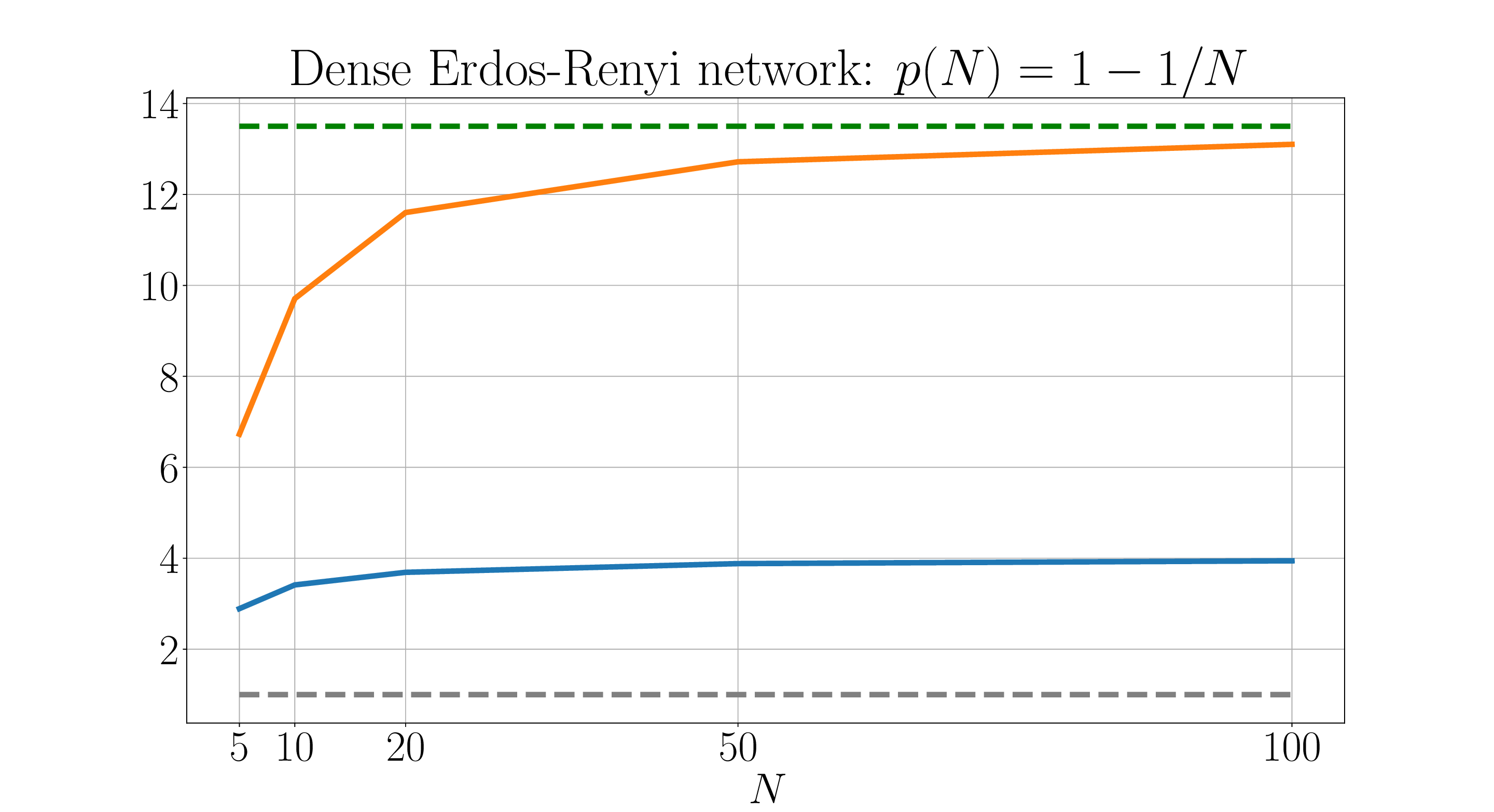}
        \end{subfigure}
        \hfill
        \begin{subfigure}[t]{0.5\textwidth}
            \includegraphics[width=\linewidth]{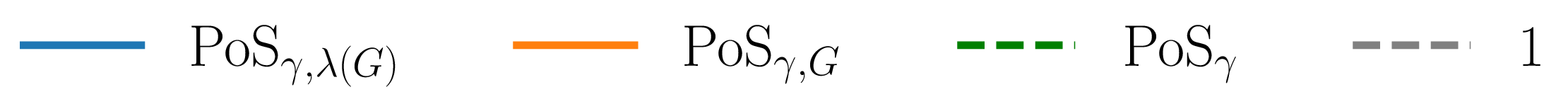}
        \end{subfigure}
        \caption{A plot of the welfare suboptimality bounds from Theorem~\ref{thm: pos} for graphs sampled randomly using the \erdosrenyi model in Example \ref{ex: dense erdos renyi networks} with $p(\numplayers) = 1/\numplayers$ and $p(\numplayers) = 1 - 1/\numplayers$ respectively, $G_{ij} \in \{0,1/\numplayers\}$, $\numplayers \in \{5, 10, 20, 50, 100\}$ and $\gamma = 1/4$. The results are averaged over $100$ instances of networks realized from each model. The bound $\posboundeigvals$ is the sharpest for both the networks. The bound $\posboundnetworkvals$ is looser but can be inferred from the network edge weights $\{G_{ij}\}_{(i,j) \in [\numplayers]^2}$ alone, and does not require computing the eigenvalues of $\adjacency$. The bound $\posboundgamma$ is the weakest since it does not use any information about the network $\adjacency$. Note that for $p(\numplayers) = 1/\numplayers$, as $\numplayers \to \infty$, the contribution of the local aggregate term $z_i(\actionvector_{-i}) = \sum_{j = 1}^\numplayers G_{ij} a_j$ to the utility of each player~$i$ in \eqref{eq: LQ utility} is negligible. Hence, it is expected that the welfare suboptimality approaches $1$. }
        \label{fig: Pos ER graphs}
    \end{figure}
\end{example}

\subsection{Additional numerical results}
\label{sec: appendix additional numerical results}

Consider the setup of Section \ref{sec: numerical results}.  
To verify that the learned actions are good approximate Nash equilibria, we compute the utility suboptimalities as shown in Table \ref{tab: utility suboptimality}. It is observed that the learned actions provide a good utility suboptimality value for the players, since they are $2\alpha$-Nash equilibria. 
\begin{table}[t]
    \centering
    \renewcommand{\arraystretch}{1.2} 
    \setlength{\tabcolsep}{2pt}
    \begin{tabular}{|c|c|c|} \hline
         \textbf{Action profile} &\textbf{Utility suboptimality} \\ \hline
         Algorithm~\ref{alg: sequential best response} &$0.01$ \\ \hline
         Algorithm~\ref{alg: gradient play} &$0.35$ \\ \hline
    \end{tabular}
    \caption{Utility suboptimality of the actions learned by Algorithms~\ref{alg: sequential best response} and \ref{alg: gradient play} for the LQ network game in Section \ref{sec: numerical results}. Given an action profile $\actionvector \in \actionset$, its utility suboptimality is defined as $\epsilon(\actionvector) := \max_{i \in [\numplayers]} [\max_{a_i' \in \actionset_i} \utility_\playeridx(a_i', \actionvector_{-i}) - \utility_\playeridx(\actionvector)] / |\utility_\playeridx(\actionvector)|$. As guaranteed by Theorems \ref{thm: sequential best response convergence} and \ref{thm: gradient play convergence}, $\epsilon(\actionvector)$ is small for an action profile $\actionvector$ learned by Algorithms~\ref{alg: sequential best response} and \ref{alg: gradient play}. 
    The reported values are averaged over $T = 100$ trials of the game. }
    \label{tab: utility suboptimality}
\end{table}

Finally, we plot the value of the $\alpha$-potential function in \eqref{eq: LQ alpha potential function} of the LQ network game as a function of the sequence of actions learned by the players using Algorithms~\ref{alg: sequential best response} and \ref{alg: gradient play} for some specific instances of the random network game. The plots are shown in Fig. \ref{fig: alpha potential for algorithms}. It can be observed that for each algorithm, the value of the $\alpha$-potential function monotonically increases. 
\begin{figure}[t]
    \centering
    \begin{subfigure}[b]{0.24\textwidth}
        \includegraphics[width=\linewidth]{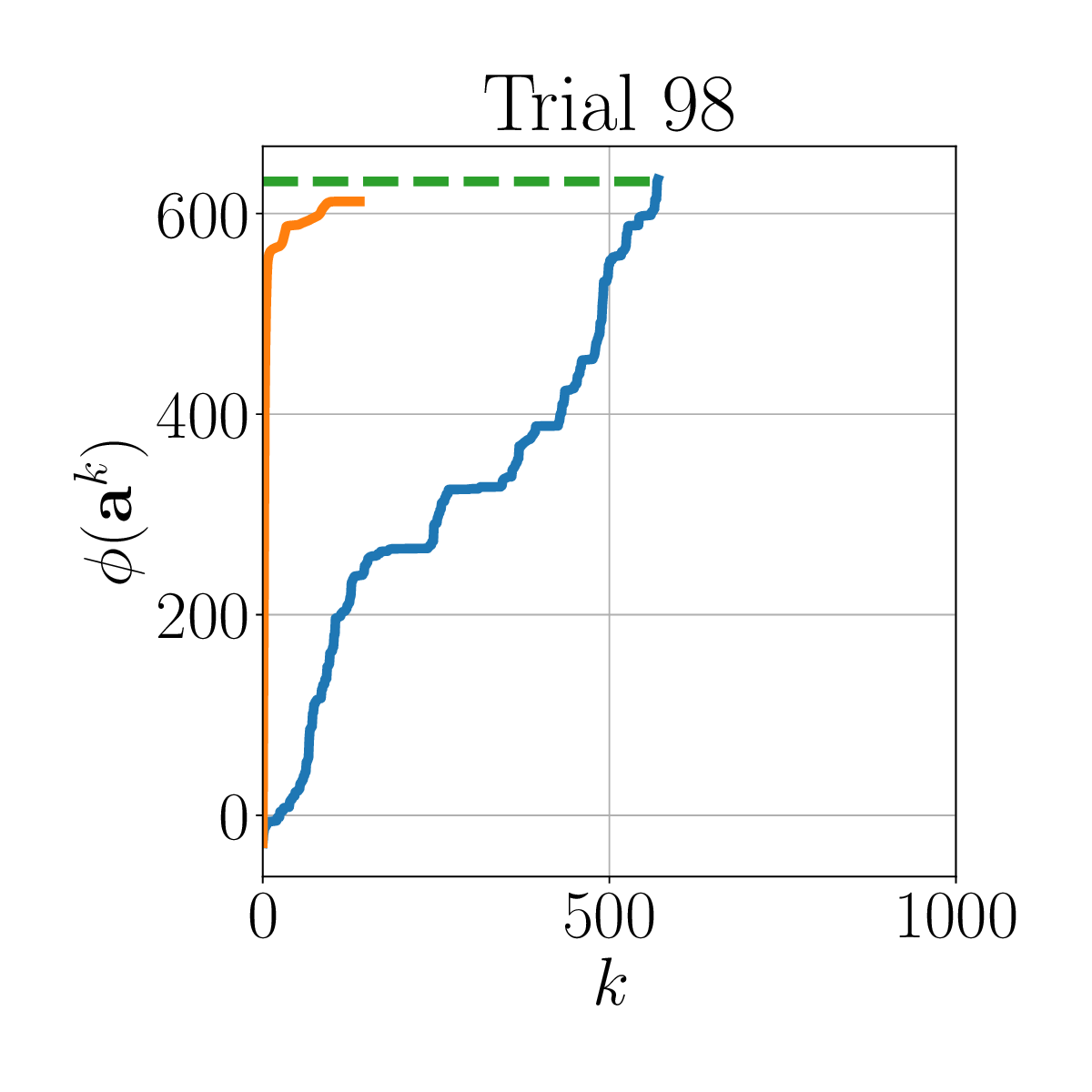}
    \end{subfigure}
    \begin{subfigure}[b]{0.24\textwidth}
        \includegraphics[width=\linewidth]{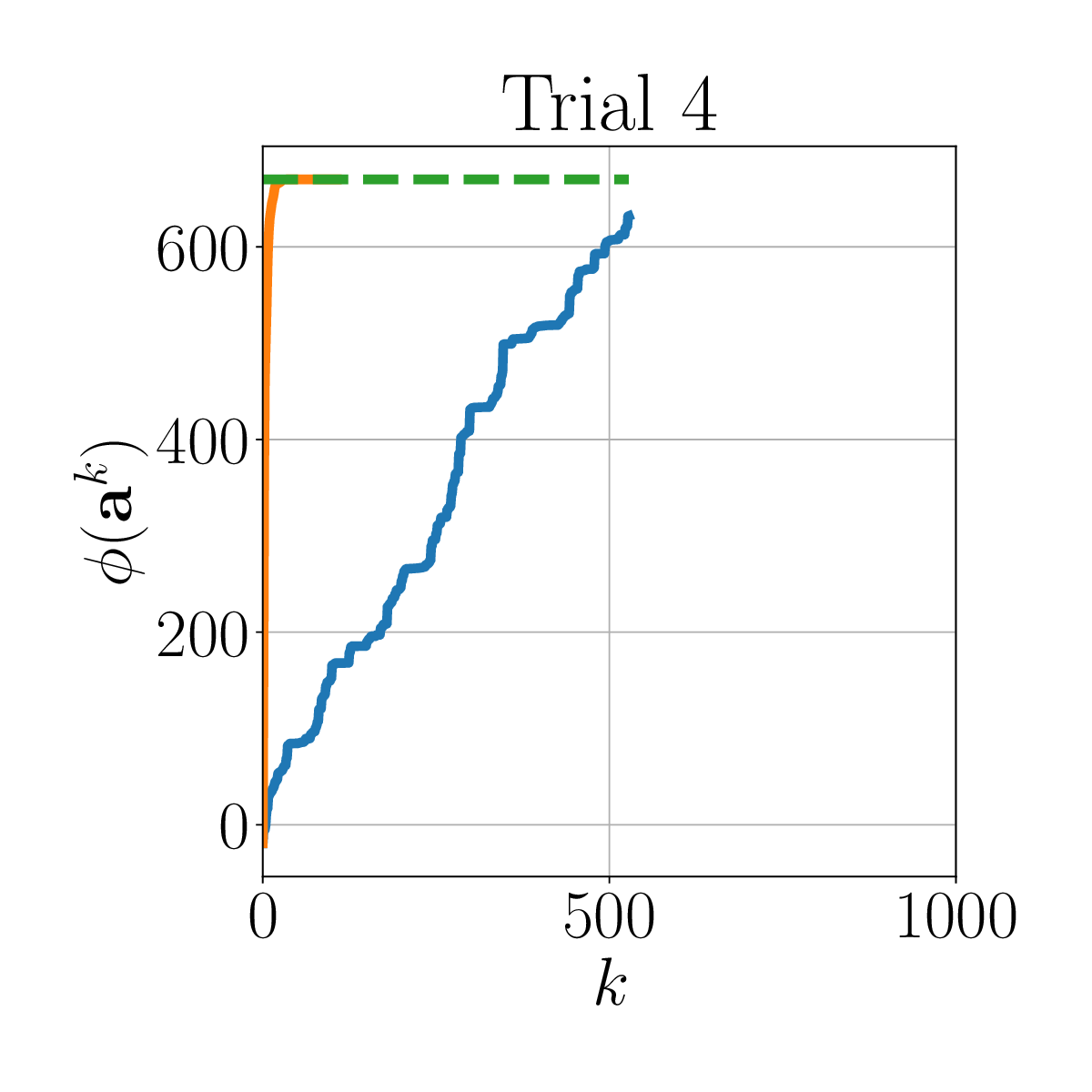}
    \end{subfigure}
    \begin{subfigure}[b]{0.24\textwidth}
        \includegraphics[width=\linewidth]{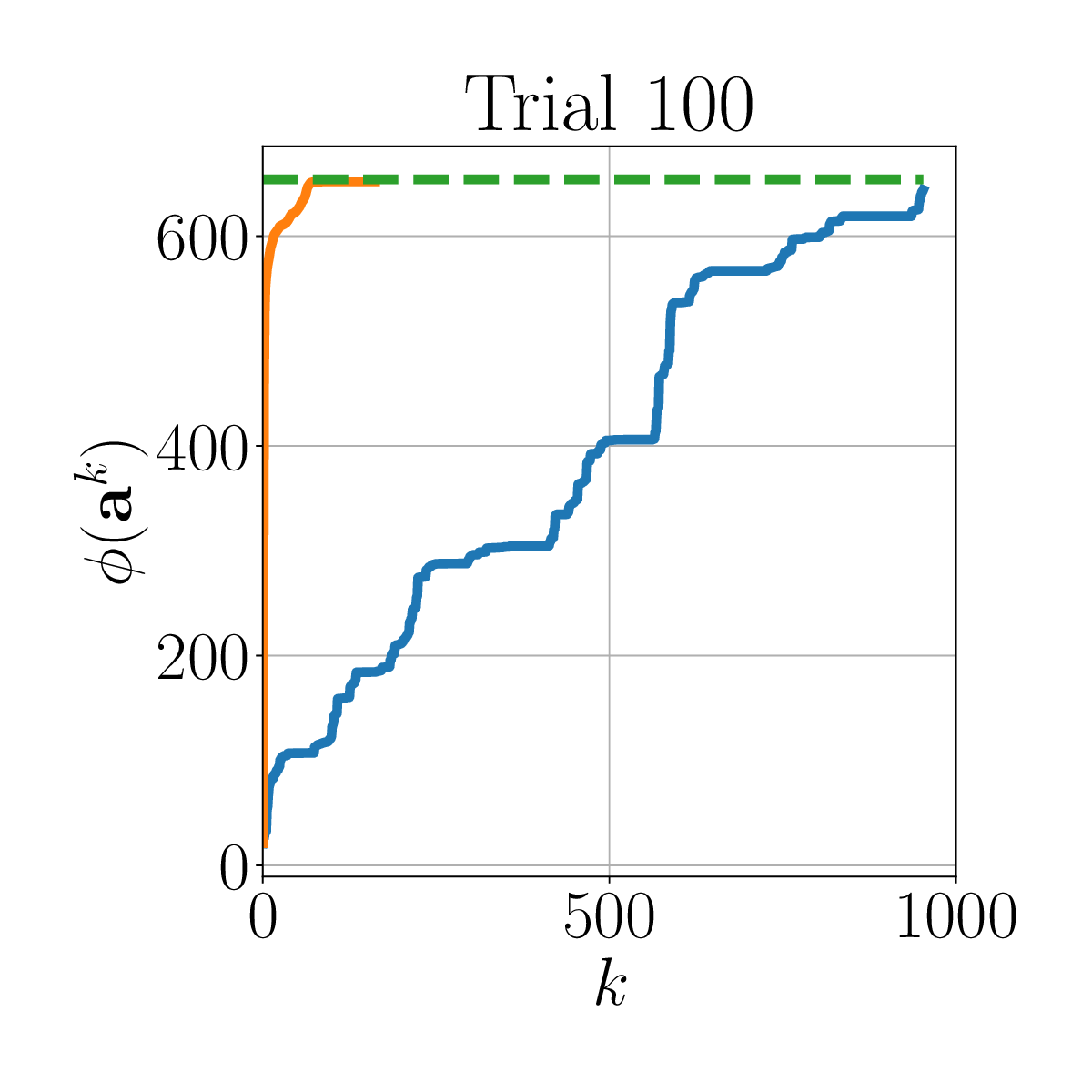}
    \end{subfigure}
    \begin{subfigure}[b]{0.24\textwidth}
        \includegraphics[width=\linewidth]{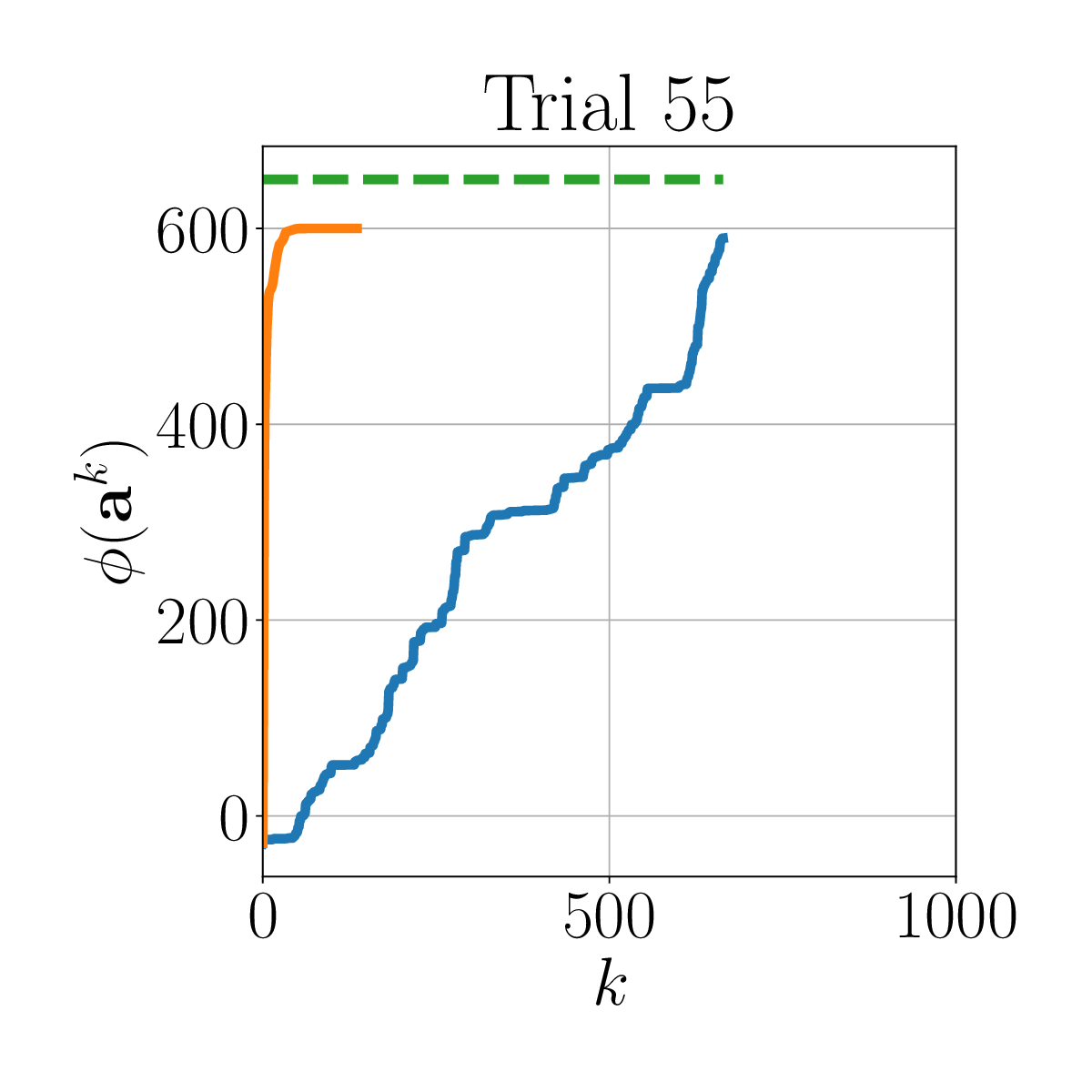}
    \end{subfigure}
    \begin{subfigure}[b]{0.5\textwidth}
        \includegraphics[width=\linewidth]{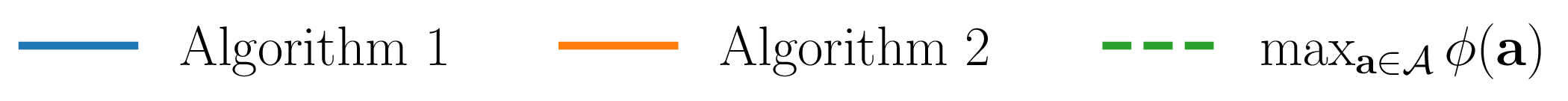}
    \end{subfigure}
    \caption{A plot of the $\alpha$-potential function in \eqref{eq: LQ alpha potential function} evaluated at each iteration of Algorithms~\ref{alg: sequential best response} and \ref{alg: gradient play} for some chosen trials of the LQ network game in Section \ref{sec: numerical results}. In each case, the $\alpha$-potential function monotonically increases along the trajectories of each algorithm. It does not always reach its maximum value. }
    \label{fig: alpha potential for algorithms}
\end{figure}

\end{document}